\documentclass[10pt]{article}
\usepackage{lettrine}
\usepackage[
  style=numeric-comp,
  sorting=none,
  sortcites=true,
  maxbibnames=99,
  giveninits=true,
  doi=true,
  url=true,
  isbn=false,
  eprint=true
]{biblatex}
\DeclareFieldFormat{doi}{\href{https://doi.org/#1}{\nolinkurl{https://doi.org/#1}}}
\usepackage[svgnames]{xcolor}
\usepackage{soul}
\usepackage{color}
\usepackage[breakable]{tcolorbox}
\usepackage{amssymb}
\usepackage{bm}
\usepackage{parskip} 
\usepackage{iftex}
\usepackage{framed}
\usepackage{tikz}
\usetikzlibrary{calc}
\ifPDFTeX
    \usepackage[T1]{fontenc}
    \usepackage{mathpazo}
\else
    \usepackage{fontspec}
\fi
\usepackage{fancyhdr}
\usepackage{caption}
\usepackage{subcaption}
\usepackage{float}
\usepackage{xcolor} 
\usepackage{enumerate} 
\usepackage{geometry} 
\usepackage{amssymb} 
\usepackage{textcomp} 
\usepackage{amsmath, amssymb, amsfonts}
\usepackage{amsthm}
\usepackage{graphicx}
\usepackage{listings}
\usepackage[hyperfootnotes=false]{hyperref}
\usepackage{color}
\usepackage{booktabs}   
\usepackage{makecell}
\usepackage{multirow}   
\definecolor{gray}{rgb}{0.5,0.5,0.5}
\AtBeginDocument{%
}
\usepackage{upquote} 
\usepackage{eurosym} 
\usepackage{fancyvrb} 
\usepackage{grffile} 
\makeatletter
\@ifpackagelater{grffile}{2019/11/01}
{}
{
  \def\Gread@@xetex#1{%
    \IfFileExists{"\Gin@base".bb}%
    {\Gread@eps{\Gin@base.bb}}%
    {\Gread@@xetex@aux#1}%
  }
}
\makeatother
\usepackage[Export]{adjustbox} 
\adjustboxset{max size={0.9\linewidth}{0.9\paperheight}}

\usepackage{titling}
\usepackage{longtable} 
\usepackage{booktabs}  
\usepackage[inline]{enumitem} 
\usepackage[normalem]{ulem} 
\usepackage{mathrsfs}

\usepackage{authblk}

\newlength{\affildateskip}
\makeatletter
\renewcommand\@maketitle{%
  \newpage
  \null
  \vskip 2em%
  \begin{center}%
  \let \footnote \thanks
  {\LARGE \@title \par}%
  \vskip 1.5em%
  {\large
   \lineskip .5em%
   \begin{tabular}[t]{c}%
     \@author
   \end{tabular}\par}%
  \vskip\affildateskip
  {\large \@date}%
  \end{center}%
  \par
  \vskip 1.5em}
\makeatother
\usepackage{chngcntr}                  
\usepackage{xcolor}
\usepackage{amssymb}
\newcommand{\benchpass}{\textcolor{green!50!black}{\checkmark}}
\newcommand{\benchfail}{\textcolor{red}{\textbf{\texttimes}}}
\newtheorem{theorem}{Theorem}[section]
\newtheorem{proposition}[theorem]{Proposition}

\theoremstyle{definition}

\definecolor{urlcolor}{rgb}{0,.145,.698}
\definecolor{linkcolor}{rgb}{.71,0.21,0.01}
\definecolor{citecolor}{rgb}{.12,.54,.11}
\hypersetup{
    breaklinks=true,
    colorlinks=true,
    urlcolor=urlcolor,
    linkcolor=linkcolor,
    citecolor=citecolor,
}

\definecolor{ansi-black}{HTML}{3E424D}
\definecolor{ansi-black-intense}{HTML}{282C36}
\definecolor{ansi-red}{HTML}{E75C58}
\definecolor{ansi-red-intense}{HTML}{B22B31}
\definecolor{ansi-green}{HTML}{00A250}
\definecolor{ansi-green-intense}{HTML}{007427}
\definecolor{ansi-yellow}{HTML}{DDB62B}
\definecolor{ansi-yellow-intense}{HTML}{B27D12}
\definecolor{ansi-blue}{HTML}{208FFB}
\definecolor{ansi-blue-intense}{HTML}{0065CA}
\definecolor{ansi-magenta}{HTML}{D160C4}
\definecolor{ansi-magenta-intense}{HTML}{A03196}
\definecolor{ansi-cyan}{HTML}{60C6C8}
\definecolor{ansi-cyan-intense}{HTML}{258F8F}
\definecolor{ansi-white}{HTML}{C5C1B4}
\definecolor{ansi-white-intense}{HTML}{A1A6B2}
\definecolor{ansi-default-inverse-fg}{HTML}{FFFFFF}
\definecolor{ansi-default-inverse-bg}{HTML}{000000}
\definecolor{outerrorbackground}{HTML}{FFDFDF}
\definecolor{darkgreen}{rgb}{0.0, 0.5, 0.0}

\DefineVerbatimEnvironment{Highlighting}{Verbatim}{commandchars=\\\{\}}

\makeatletter
\newcommand{\boxspacing}{\kern\kvtcb@left@rule\kern\kvtcb@boxsep}
\makeatother

\title{Scalable dynamical inference of phase-field fracture from sparse and partial measurements}
\author[1]{Hanfeng Zhai\thanks{Corresponding author. E-mail: \tt hzhai@stanford.edu}}
\author[2]{Zisheng Zhang\thanks{E-mail: \tt zishengz@stanford.edu}}
\affil[1]{Department of Mechanical Engineering, Stanford University, Stanford, CA 94305, USA}
\affil[2]{Department of Chemical Engineering, Stanford University, Stanford, CA 94305, USA}

\date{\today}

\begin{document}
\maketitle

\begin{abstract}

In structural health monitoring and fracture assessment, evolving crack fields must often be inferred from sparse mechanical measurements rather than dense full-field observations.
We develop a convolutional-recurrent framework (CNN2D--ConvGRU) for measurement-conditioned reconstruction of time-dependent phase-field brittle fracture.
The framework learns an approximate, observation-conditioned update map for the regularized fracture evolution.
At each load step, it maps a fixed-length history of phase and displacement fields, together with sparse displacement measurements at the current step, to the current full-field phase and displacement state.
During sequential deployment, new measurements are assimilated at every load step. The framework therefore serves as a state-inference surrogate rather than an autonomous time integrator.
%
%
Across the held-out cases, the reconstructed fields reproduce the crack paths, damage evolution, and bulk displacement response.
The largest discrepancies remain localized near propagating crack tips and steep displacement gradients, while some drift accumulates during late-stage sequential reconstruction.
The trained weights are further evaluated without retraining on a $512 \times 512$ raster after training at $256 \times 256$, using the same physical domain and finite-element discretization with a proportionally refined measurement grid.
For the tested crack configurations, this empirical raster-and-sensing transfer preserves the principal damage topology and global damage evolution, although fine-scale displacement errors increase near crack tips.
Controlled comparisons with alternative spatial and temporal architectures show that CNN2D--ConvGRU provides a favorable balance between reconstruction accuracy and computational cost for the problem considered.
Compared with repeated finite-element solutions, sequential full-field reconstruction achieves mean speedups of $175\times$ on CPU and $253\times$ on GPU.
These results demonstrate that convolutional sequence learning can efficiently reconstruct evolving full-field fracture states from sparse and partial observations while retaining the spatial structure and history dependence of phase-field fracture.
\vspace{8pt}

\end{abstract}

\clearpage

\section{Introduction}\label{sec:introduction}

Failure and fracture govern stiffness loss and ultimate load capacity across structural alloys, ceramics, composites, and thin-film systems.
Predicting crack nucleation, propagation, and coalescence therefore remains a central challenge in solid mechanics because failure is often controlled by evolving discontinuities rather than by bulk strength alone \cite{griffith1921,irwin1957,rice1968,dugdale1960,barenblatt1962}.
Crack growth is a free-discontinuity problem: paths may branch or merge, local tip fields couple to global constraint, and heterogeneous or dynamic settings add further history dependence \cite{rice1968,spatschek2011,bouchbinder2014}.
Atomistic and atomistic-to-continuum models resolve bond breaking and tip-scale mechanisms but remain costly except in localized nonlinear zones \cite{buehler2006,tadmor1996,shenoy1999}.

At the engineering scale, finite element methods (FEM) dominate continuum analysis.
Explicit crack tracking, cohesive-zone laws with prescribed paths, and extended finite element enrichments each carry practical limitations when topology evolves \cite{xu1994,moes1999}.
Phase-field fracture regularizes Griffith's energetic picture by replacing sharp cracks with a smooth damage-like field on a fixed mesh \cite{francfort1998,bourdin2000,bourdin2008}.
Governing equations follow from stationarity of a degraded elastic energy plus a fracture-surface functional, with irreversibility enforced to prevent healing \cite{bourdin2007,miehe2010a,miehe2010b}.
The same representation handles branching, merging, and coalescence without remeshing \cite{karma2004,spatschek2011,miehe2010a,ambati2015}.
The variational framework also extends to coupled problems such as hydraulic fracture, hydrogen-assisted cracking, and fatigue \cite{lee2017,martinez2018,carrara2019fatigue}.
In practice, the coupled displacement--phase-field problem is solved with staggered finite element updates on a fixed background mesh \cite{miehe2010a,miehe2010b}.
Irreversibility and history-dependent internal variables make crack evolution path dependent, so each load increment requires a constrained nonlinear solve rather than a single static boundary-value problem \cite{gerasimov2019,miehe2010a}.
Because cracks are represented implicitly, the method scales to interacting two- and three-dimensional networks without remeshing when branches merge or coalesce \cite{karma2004,ambati2015}.
These features have driven adoption in problems where topology changes, constraint effects, and coupled driving fields must be treated together.

These strengths come with substantial computational cost.
The regularization length must be well resolved, which demands fine meshes near diffuse crack bands and large nonlinear systems on long load paths \cite{ambati2015,bourdin2007}.
Irreversible, path-dependent evolution requires repeated staggered or monolithic solves over many load increments \cite{miehe2010a,gerasimov2019}.
Predictions can be sensitive to degradation, tension--compression splitting, and nucleation modeling \cite{tanne2018,sargado2017,kumar2020}.
Adaptive mesh refinement concentrates degrees of freedom near evolving crack bands \cite{heister2015,Hirshikesh2019}.
Hybrid and operator-split schemes improve robustness under compression while retaining staggered efficiency \cite{ambati2015,miehe2010a}.
Revised degradation laws and strength-aware variants sharpen nucleation and peak-load predictions \cite{sargado2017,tanne2018,kumar2020}.
Monolithic and quasi-Newton solvers with adaptive load stepping further accelerate difficult fracture simulations \cite{farrell2017,kristensen2020}.
However, existing phase-field fracture methods remain too computationally expensive for applications requiring repeated, time-dependent predictions, creating a need for an efficient surrogate capable of accurately reconstructing evolving crack fields.

Therefore, we develop a surrogate model for time-dependent crack propagation.
Data-driven models learn approximate time-update maps from simulation data without replacing the underlying physics outright.
A common aim is to approximate an expensive per-step operator or a reduced state manifold learned from high-fidelity trajectories.
In monitoring and measurement settings, an equally pressing need is to recover the evolving full-field fracture state when only sparse or partial sensors are available at the current step.
Autoencoder-based nonlinear model reduction projects high-dimensional fields onto compact latent coordinates \cite{lee2020}.
Physics-informed neural networks embed governing-equation residuals and constraints directly in the training objective \cite{raissi2019, zhai2022predicting}.
Convolutional encoder--decoder architectures predict and evolve gridded state fields, exploiting local spatial correlations and weight sharing \cite{geneva2020,lecun1998}.
Graph neural networks operate on unstructured mesh connectivity and can learn dynamics directly on simulation graphs \cite{pfaff2021, zhai2025}.
Within computational mechanics, related ideas appear in field forecasting for multiphase flow, graph surrogates for polycrystal plasticity, and operator-learning approaches to brittle fracture \cite{wen2021,zhai2025,goswami2022}.
Phase-field simulations are a natural match for this grid-based view because the damage and displacement components can be rasterized as multichannel images.
A time-dependent fracture run then becomes a measurement-conditioned inference problem: given past field snapshots and sparse sensors at the current step, reconstruct the current full-field state.
Convolutional long short-term memory (ConvLSTM) networks introduced spatiotemporal recurrence with convolutional state transitions \cite{shi2015}, and convolutional gated recurrent units (ConvGRU) provide a closely related gated update with a reduced parameter count \cite{ballas2016}.
Both architectures preserve spatial layout during temporal propagation, in contrast to recurrent models that flatten the domain before each step.

Here we propose a fully convolutional CNN2D--ConvGRU model trained on FEniCS phase-field trajectories for dynamical inference of fracture fields under sparse sensing.
Two-dimensional convolutions extract multiscale features from each frame; a ConvGRU module advances these latent fields with temporal memory.
We train separate networks for $\phi$, $u_x$, and $u_y$ using the past-history input $\mathcal{T}$ ($10$ frames).
At inference, the model receives coarse $32\times32$ sparse displacement samples and a masked phase field at the current step, reflecting the limited observations typical of structural monitoring rather than dense full-field data.
The network learns a measurement-conditioned reconstruction of the current fields rather than an autonomous time-step map that advances without current-step observations.
All learned layers are convolutional, so the same weights can be evaluated at different resolutions.
Resolution transfer is treated as an empirical test rather than a theoretical invariance.
Mask-weighted losses emphasize crack bands; test-set simulations check generalization beyond the training ensemble.
%

The paper is organized as follows.
Section~\ref{sec:methods} presents the variational phase-field reference model, the simulation ensemble and sparse sensing setup, and the CNN2D--ConvGRU model with training protocol.
Section~\ref{sec:results} reports current-step reconstruction accuracy and parity on the report set.
Section~\ref{sec:deployment} applies sequential measurement-conditioned reconstruction on report-set configurations.
Section~\ref{sec:scaling} examines sequential deployment on a finer raster representation of the same physical domain without retraining.
Section~\ref{sec:discussions} discusses raster-resolution transfer, displacement versus damage accuracy, and deployment limitations.
Section~\ref{sec:conclusions} summarizes the main findings and outlook.
Appendix~\ref{sec:app_pf_sim} tabulates phase-field simulation and data-set parameters, and Appendix~\ref{sec:app_ml_train} lists network architecture and loss hyperparameters.

\section{Methods}\label{sec:methods}

\subsection{Computational fracture mechanics}

We generate the high-fidelity training data with a variational phase-field model of brittle fracture.
A sharp crack is regularized into a continuous damage-like field $\phi(\mathbf{x},t)\in[0,1]$, where $\phi=0$ denotes intact material and $\phi=1$ a fully developed crack \cite{francfort1998,miehe2010a}.
The coupled problem is posed on a two-dimensional square domain $\Omega=[-0.5,0.5]^2$ and solved with FEniCS \cite{fenics}.
We use continuous piecewise-linear (CG1) Lagrange elements for both the displacement $\mathbf{u}$ and the phase field $\phi$.
We interpret the unit square as a physical specimen of size $1\,\mathrm{m}\times1\,\mathrm{m}$ ($L=1\,\mathrm{m}$).
Displacements in mm and the diffuse-band length scales below are referenced to this size.
The FEniCS mesh uses $300$ cells per side ($\approx 3.3\,\mathrm{mm}$ elements).
This is the solver mesh, not the later ML raster ($256\times256$ or $512\times512$).
The regularization length $\ell=10^{-2}$ is $10\,\mathrm{mm}$, so a diffuse crack band spans a few elements.
The total energy combines an elastic contribution degraded by damage and a fracture surface energy of the Ambrosio--Tortorelli (AT2) type,
\begin{equation}
    \Psi(\mathbf{u},\phi)
    = \int_{\Omega} g(\phi)\,\psi(\mathbf{u})\,\mathrm{d}\Omega
    + G_c \int_{\Omega} \left( \frac{1}{2\ell}\,\phi^2 + \frac{\ell}{2}\,|\nabla\phi|^2 \right) \mathrm{d}\Omega ,
    \label{eq:total_energy}
\end{equation}
where $G_c$ is the critical fracture energy, $\ell$ is the regularization length scale that controls the width of the diffuse crack band, and $g(\phi)=(1-\phi)^2$ is the quadratic degradation function used in the finite-element implementation.

The material is modeled as linear isotropic elastic. The infinitesimal strain and Cauchy stress are
\begin{equation}
    \boldsymbol{\varepsilon}(\mathbf{u}) = \tfrac{1}{2}\!\left(\nabla\mathbf{u} + \nabla\mathbf{u}^{\top}\right),
    \qquad
    \boldsymbol{\sigma}(\mathbf{u}) = 2\mu\,\boldsymbol{\varepsilon} + \lambda\,\mathrm{tr}(\boldsymbol{\varepsilon})\,\mathbf{I},
    \label{eq:constitutive}
\end{equation}
with Lam\'e parameters $\lambda$ and $\mu$. To suppress damage driven by negative volumetric strain, the crack-driving energy uses a volumetric--deviatoric split with only the positive volumetric contribution retained,
\begin{equation}
    \psi(\mathbf{u}) = \tfrac{1}{2}(\lambda+\mu)\,\langle \mathrm{tr}\,\boldsymbol{\varepsilon}\rangle_{+}^{2}
    + \mu\,\mathrm{dev}(\boldsymbol{\varepsilon})\!:\!\mathrm{dev}(\boldsymbol{\varepsilon}),
    \label{eq:psi_split}
\end{equation}
where $\langle x\rangle_{+}=\tfrac{1}{2}(x+|x|)$ denotes the positive part and $\mathrm{dev}(\cdot)$ the deviatoric operator.
Irreversibility is enforced through a history field that records the maximum strain energy density $\psi(\mathbf{u})$ from Eq.~\eqref{eq:psi_split} over the loading history,
\begin{equation}
    H(\mathbf{x},t) = \max_{\tau \le t} \psi\big(\mathbf{u}(\mathbf{x},\tau)\big),
    \label{eq:history}
\end{equation}
which is the standard history-field approximation used to suppress crack healing in phase-field fracture calculations.
We do not claim that this construction is mathematically equivalent to the pointwise constraint $\phi_{n+1}\ge\phi_n$ or to an explicit projection $\phi_{n+1}=\max(\phi_n,\widetilde{\phi}_{n+1})$.

We use a hybrid history-field phase-field formulation.
Mechanical equilibrium is solved using the degraded elastic stress.
Crack evolution is driven by the nondecreasing history field in Eq.~\eqref{eq:history}.
Because the history-field substitution is not derived here from a single constrained incremental potential, Eqs.~\eqref{eq:weak_u}--\eqref{eq:weak_phi} are the governing equations of the finite-element implementation.
They are not presented as simultaneous Euler--Lagrange equations of Eq.~\eqref{eq:total_energy}.
In weak form, with $\boldsymbol{\sigma}(\mathbf{u})$ from Eq.~\eqref{eq:constitutive} and $g(\phi)=(1-\phi)^2$, the displacement subproblem reads: find $\mathbf{u}$ such that
\begin{equation}
    \int_{\Omega} g(\phi)\,\boldsymbol{\sigma}(\mathbf{u}):\boldsymbol{\varepsilon}(\mathbf{v})\,\mathrm{d}\Omega = 0
    \qquad \forall\,\mathbf{v},
    \label{eq:weak_u}
\end{equation}
while the phase-field subproblem reads: find $\phi$ such that
\begin{equation}
    \int_{\Omega} \Big[ G_c\,\ell\,\nabla\phi\!\cdot\!\nabla q
    + \big(\tfrac{G_c}{\ell} + 2H\big)\,\phi\,q
    - 2H\,q \Big]\,\mathrm{d}\Omega = 0
    \qquad \forall\,q .
    \label{eq:weak_phi}
\end{equation}
Let $\Gamma_{\mathrm{bottom}}$ and $\Gamma_{\mathrm{top}}$ denote the bottom
and top boundaries. At load step $t$, the displacement belongs to the affine
space
\[
\mathcal{V}_t
=
\left\{
\mathbf{u}\in[H^1(\Omega)]^2:
\mathbf{u}=\mathbf{0}
\ \text{on }\Gamma_{\mathrm{bottom}},
\quad
u_y=t\,u_r
\ \text{on }\Gamma_{\mathrm{top}}
\right\},
\]
while the displacement test functions belong to
\[
\mathcal{V}_0
=
\left\{
\mathbf{v}\in[H^1(\Omega)]^2:
\mathbf{v}=\mathbf{0}
\ \text{on }\Gamma_{\mathrm{bottom}},
\quad
v_y=0
\ \text{on }\Gamma_{\mathrm{top}}
\right\}.
\]
At each quasi-static load step, we seek
$\mathbf{u}(\cdot,t)\in\mathcal{V}_t$ and
$\phi(\cdot,t)\in H^1(\Omega)$, with Eq.~\eqref{eq:weak_u} tested against all
$\mathbf{v}\in\mathcal{V}_0$.

At each load increment, the displacement and phase-field subproblems are solved successively while the coupling fields are held fixed.
After each displacement solve, the history field is updated, and the staggered iteration continues until the changes in both fields fall below the prescribed tolerance.
This fixed-point staggered procedure is used to generate the finite-element reference trajectories.
Because the history-field formulation is not written here as a constrained minimization of a single common incremental energy, we do not claim global convergence or monotone descent of the total energy in Eq.~\eqref{eq:total_energy}.
The implemented $L^2$ tolerance of $10^{-3}$ terminates the staggered iteration in practice when successive field updates become sufficiently small.

Loading is applied as monotonic uniaxial tension.
The bottom edge is fully clamped, $\mathbf{u}=\mathbf{0}$.
The vertical displacement on the top edge is ramped as $u_y = t\,u_r$, with the remaining components left traction-free.
The pseudo-time increment $\Delta t$ is adaptive and is reduced when the staggered iteration converges slowly.
We take $u_r=5\times10^{-3}\,\mathrm{m}=5\,\mathrm{mm}$ (top-edge displacement at $t=1$) and a base step $\Delta t=5\times10^{-2}$, reduced to $10^{-3}$ after the first increment and further near instabilities.
In the steady regime each converged frame advances the applied top displacement by $\Delta u_y=\Delta t\,u_r\approx5\times10^{-3}\,\mathrm{mm}$ ($5\,\mu\mathrm{m}$); the first step is $0.25\,\mathrm{mm}$.
Over the full history the top edge reaches $u_y=t\,u_r=2.5\,\mathrm{mm}$, a nominal strain $\varepsilon=u_y/L=0.25\%$.
Because the model is quasi-static, this load parameter, rather than a physical clock, sets the time axis.
Pre-existing cracks are imposed as Dirichlet constraints $\phi=1$ on the prescribed crack subdomains.
At each converged step we record a scalar top-boundary force diagnostic
\begin{equation}
    F_y = \int_{\Gamma_{\mathrm{top}}}
    \bigl[\boldsymbol{\sigma}(\mathbf{u})\,\mathbf{n}\bigr]_y\,\mathrm{d}s ,
    \label{eq:reaction}
\end{equation}
computed from the undegraded Cauchy traction of the staggered displacement field.
This definition matches the finite-element logging used to build Fig.~\ref{fig:schematic_Fu}(b) and the termination rule $F_y<100$ for ten consecutive converged steps.
It is \emph{not} the variational reaction associated with the degraded equilibrium form Eq.~\eqref{eq:weak_u}, which would instead use $g(\phi)\,\boldsymbol{\sigma}(\mathbf{u})\,\mathbf{n}$.
We therefore treat $F_y$ only as an implementation-defined load-path monitor and stopping criterion, not as the dual force of Eq.~\eqref{eq:weak_u}.
The mesh resolution, material constants ($G_c$, $\ell$, $\lambda$, $\mu$), loading ramp, solver tolerances, and randomized crack-ensemble rules used to build the training corpus are listed in Appendix~\ref{sec:app_pf_sim} (Table~\ref{tab:pf_sim}).
Figure~\ref{fig:schematic_Fu} summarizes the benchmark geometry, loading protocol, and the resulting force-diagnostic--displacement curve from Eq.~\eqref{eq:reaction}.

\begin{figure}[htbp]
    \centering
    \begin{subfigure}{.35\linewidth}
        \centering
        \includegraphics[width=\linewidth]{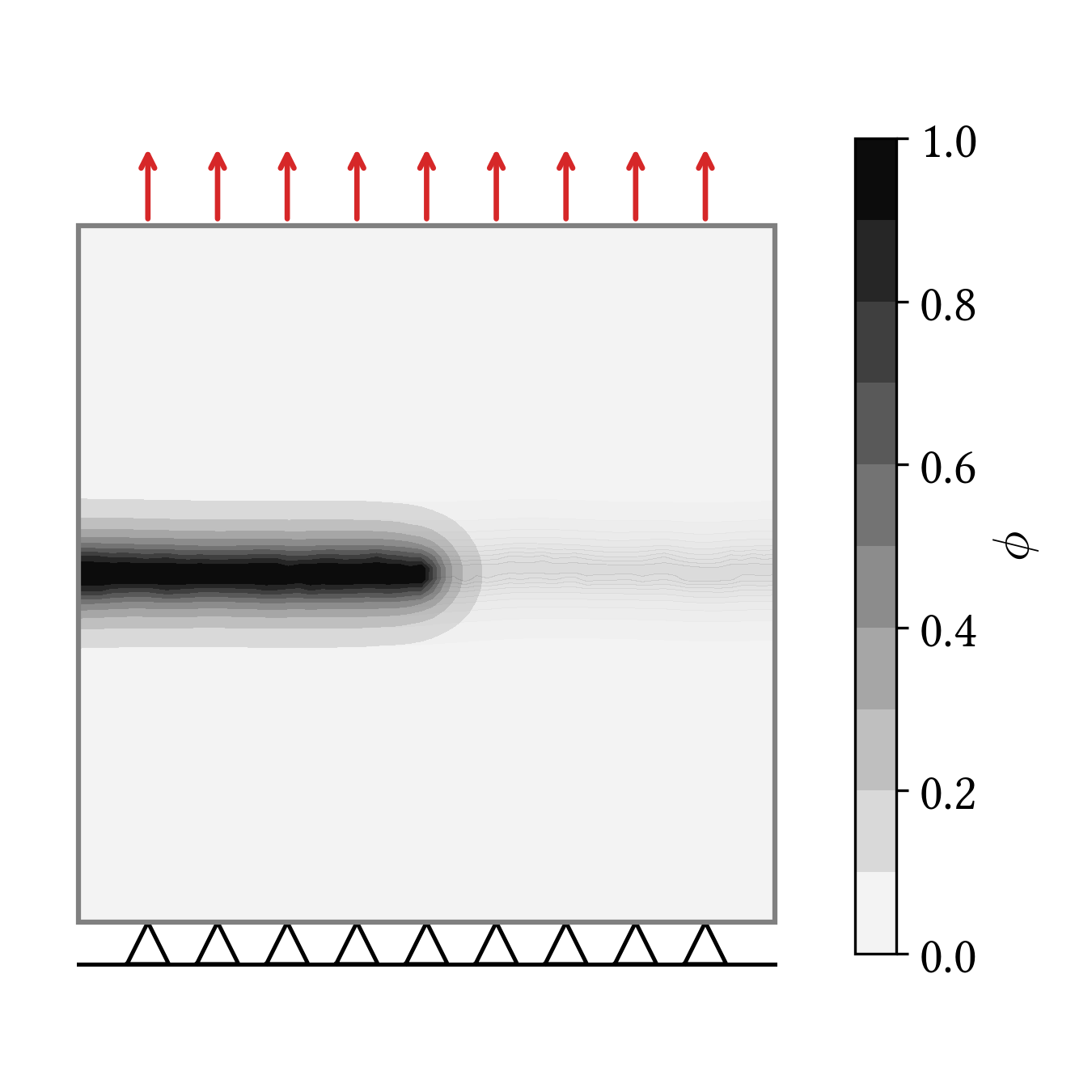}
        \label{fig:schematic}\caption{}
    \end{subfigure}
    \begin{subfigure}{.43\linewidth}
        \centering
        \includegraphics[width=\linewidth]{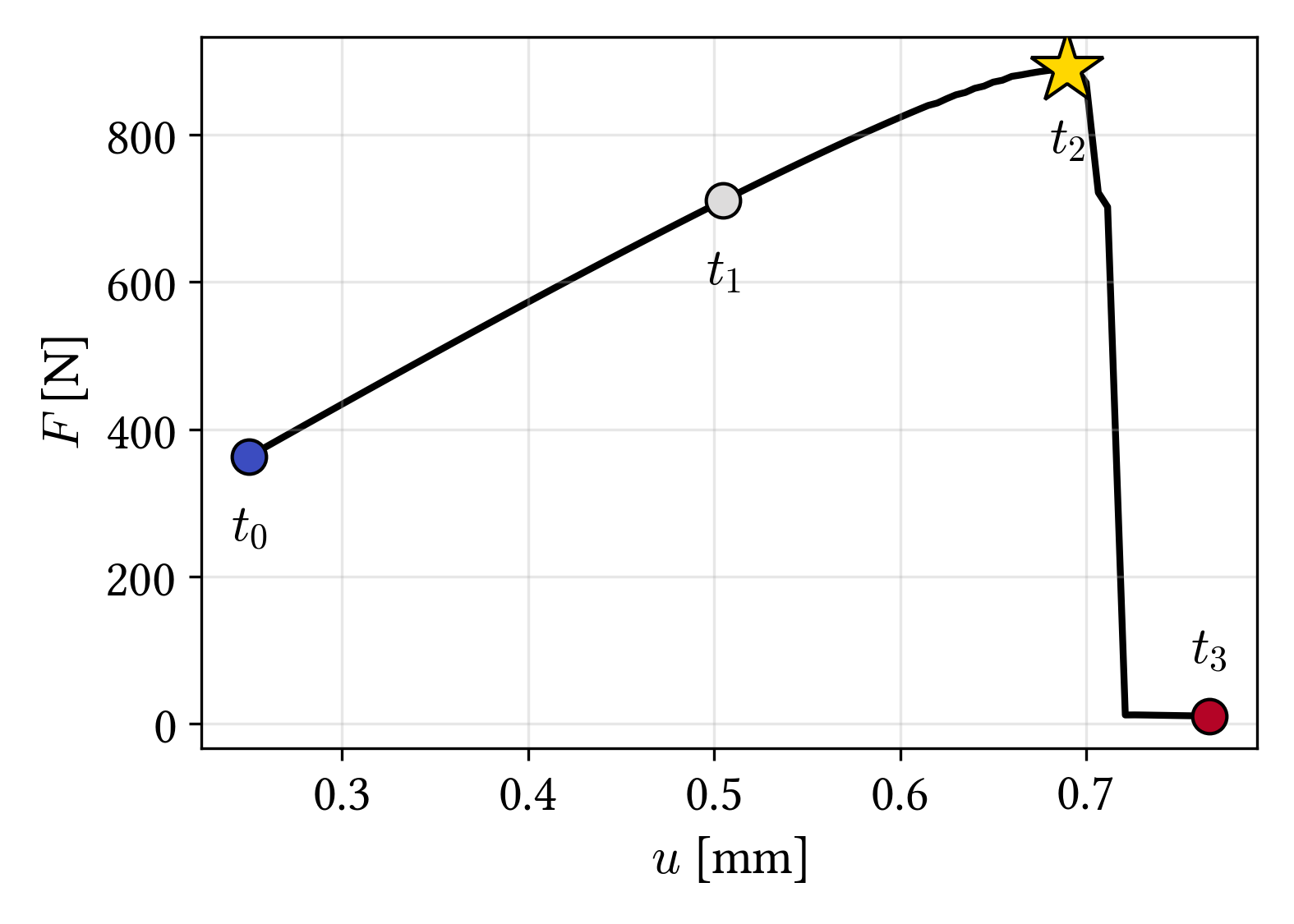}
        \label{fig:Fu_curve}\caption{}
    \end{subfigure}
    \caption{Phase-field fracture benchmark setup and implementation-defined load-path diagnostic used to generate the training trajectories. (a) Computational domain and loading protocol (schematic overlay, not a single saved snapshot); (b) top-boundary force diagnostic from Eq.~\eqref{eq:reaction} versus applied displacement, with visualization markers $t_0$--$t_3$.}
    \label{fig:schematic_Fu}
\end{figure}

In Figure~\ref{fig:schematic_Fu}(a), pre-existing cracks and the phase-field damage are drawn together in one schematic to illustrate the setup; each converged save in the computation stores a single $\phi$ field, not a composite image.
Figure~\ref{fig:schematic_Fu}(b) plots the force diagnostic from Eq.~\eqref{eq:reaction} versus pseudo-time, equivalently the applied top displacement $u_y=t\,u_r\in[0,2.5]\,\mathrm{mm}$.
The marked instants $t_0$, $t_1$, $t_2$, and $t_3$ are four loading steps chosen along that curve; Fig.~\ref{fig:fields} shows the full $\phi$, $u_x$, and $u_y$ fields at exactly those same instants for the same representative simulation.

To build a diverse data set, we generate $100$ independent fracture trajectories by randomizing the initial crack population according to the ensemble rules in Table~\ref{tab:pf_sim}.
Each simulation is run independently with the staggered phase-field solver described above, producing a full spatiotemporal history of crack nucleation, propagation, branching, and coalescence under increasing tension.

The finite element fields are exported at every converged load step and interpolated onto a uniform Cartesian grid for convolutional learning.
Each trajectory stores $\phi$ of shape $(T,H,W)$ and $(u_x,u_y)$ of shape $(T,2,H,W)$, where $T$ is the number of retained time steps.
We keep the final portion of each history, where fracture dynamics are most active.
The ensemble is split by trajectory into $75$ training and $25$ held-out report trajectories (Appendix Table~\ref{tab:ml_train}).
Report trajectories are not used for parameter updates.
Reported $R^2$ values, checkpoints, and model-selection decisions use this held-out set.

\begin{figure}[htbp]
    \centering
    \includegraphics[width=.6\linewidth]{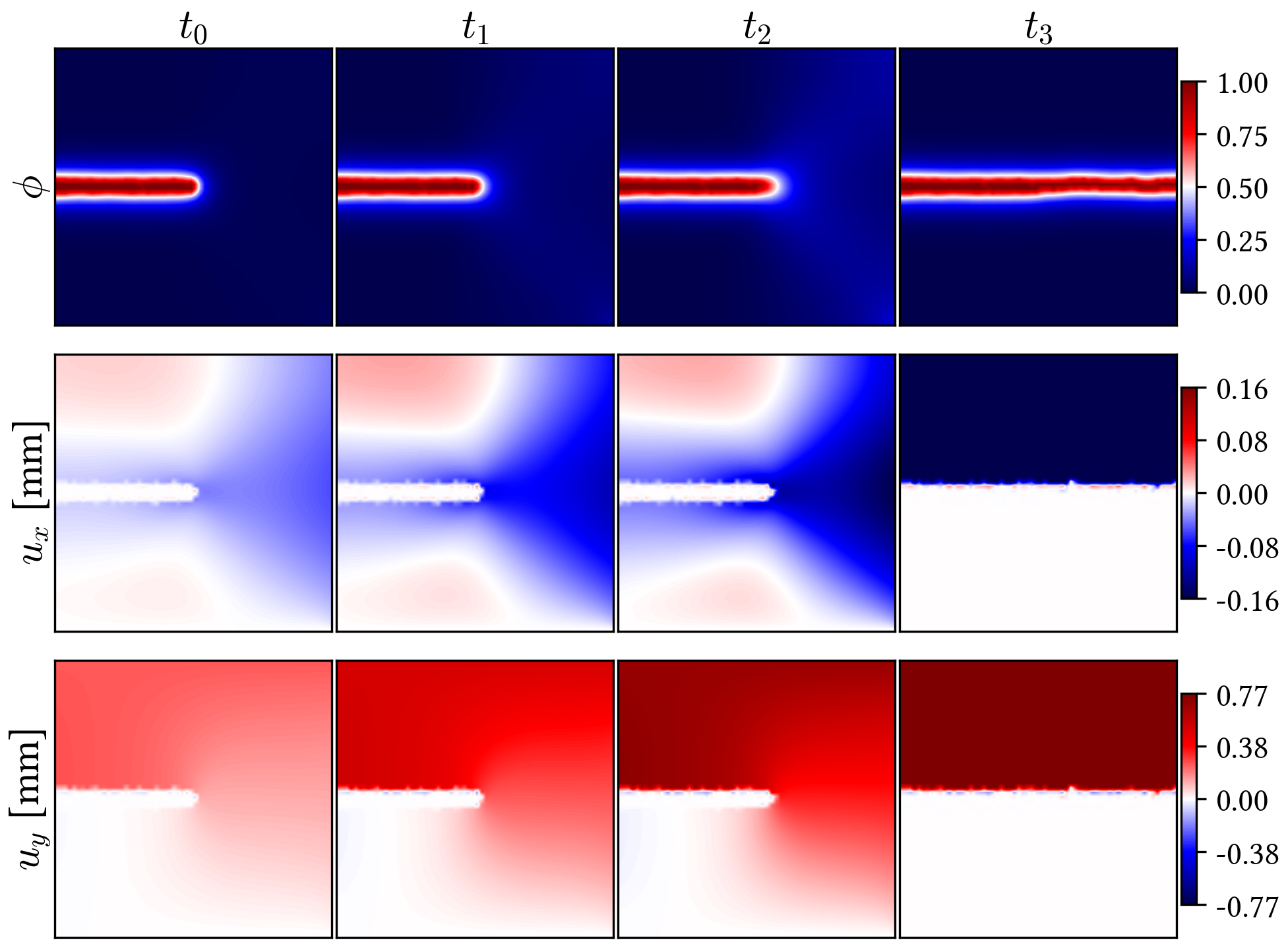}
    \caption{Representative snapshots of $\phi$, $u_x$, and $u_y$ at $t_0$--$t_3$ (marked in Figure~\ref{fig:schematic_Fu}(b)) for the same reference simulation. These spatial fields are stored as image-like channels for surrogate training.}
    \label{fig:fields}
\end{figure}

Figure~\ref{fig:fields} matches Fig.~\ref{fig:schematic_Fu} one to one.
The four columns are the fields at $t_0$, $t_1$, $t_2$, and $t_3$ on the force-diagnostic--displacement curve in Figure~\ref{fig:schematic_Fu}(b).
Column $k$ shows $\phi$, $u_x$, and $u_y$ at pseudo-time $t_k$.
Training stacks use the same per-step raster data, not the illustrative overlay in Figure~\ref{fig:schematic_Fu}(a).

\subsection{Machine learning}

We denote the rasterized state at load step $t$ by $X_t=(\phi(\cdot,t),\mathbf{u}(\cdot,t))$.
At each load step, the network reconstructs the current full-field phase and displacement state from a finite history of prior states and sparse displacement measurements available at that step.
No global regularity assumption is imposed on the underlying fracture evolution, which may be strongly nonlinear near crack initiation, branching, and coalescence.

We cast measurement-conditioned state inference as an image-based spatiotemporal regression problem. The rasterized fields are treated as multi-channel images, and the network learns a history-and-sensor-conditioned reconstruction of the \emph{current} fields.
We denote by $\mathcal{T}$ the past-history input: the $10$ consecutive prior field states $(\phi,u_x,u_y)$ at the native grid resolution (Fig.~\ref{fig:input_output}).
Given $\mathcal{T}$ and sparse measurements at the current step, the network reconstructs the current phase and displacement fields.
In sequential deployment the same reconstruction is repeated at every load step with fresh sparse measurements; the model is not an autonomous time integrator that advances without current-step sensors.

\begin{figure}[htbp]
    \centering
    \includegraphics[width=.65\linewidth]{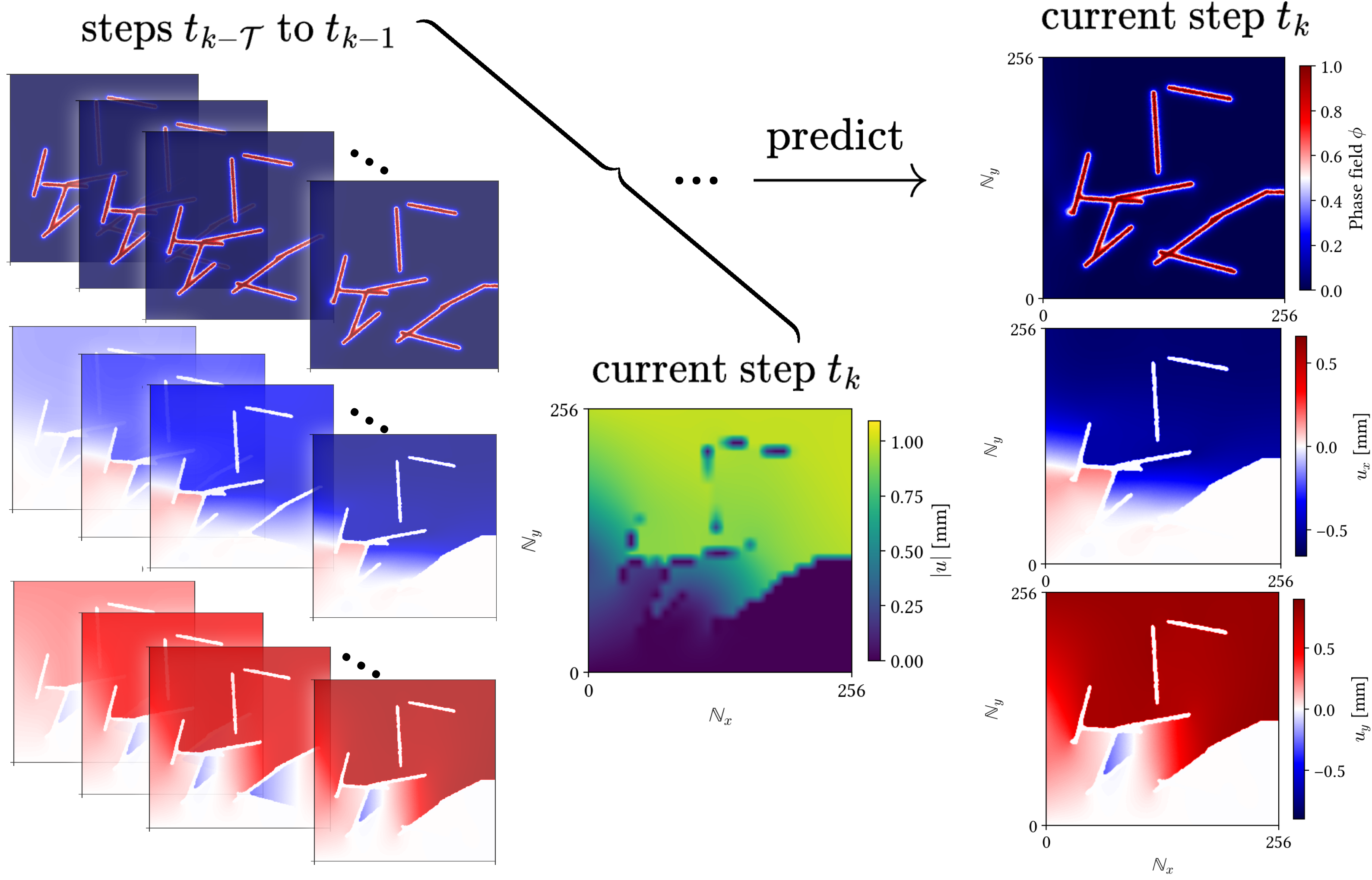}
    \caption{Spatiotemporal inference workflow. The past-history input $\mathcal{T}$ comprises prior $\phi$, $u_x$, and $u_y$ fields; together with sparse current-step displacement measurements, the CNN2D--ConvGRU model reconstructs the current full-field phase and displacement state.}
    \label{fig:input_output}
\end{figure}

Figure~\ref{fig:input_output} makes the information regime of the surrogate explicit.
The left block is the fixed-length past window $\mathcal{T}$ of full-field states, while the right block is the current step, where only a coarse displacement measurement is known and the present phase channel is withheld.
The network therefore cannot copy $\phi$ from the current frame; it must infer the current damage and displacement fields jointly from history and sparse sensing.
This layout is the same one used later for sequential deployment, with $\mathcal{T}$ rolling forward after each reconstruction.

\paragraph{Sparse and partial inputs.}
The displacement is supplied only through a coarse measurement.
At each step, the full-resolution displacement is downsampled to a $32\times32$ sensor grid and bilinearly resampled to the working resolution.
This mimics a sparse sensor array rather than dense field knowledge.
The model receives $\mathcal{T}$ and this sparse current-step sensing, while the current phase channel is masked out.
It must therefore reconstruct the complete current $\phi$, $u_x$, and $u_y$ from $\mathcal{T}$ and partial present measurements.

\paragraph{Field normalization.}
To place all channels on comparable scales, the phase field is shifted to $\phi' = \phi + 1 \in [1,2]$. Displacement components are normalized by the prescribed top-boundary displacement at load step $k$,
\begin{equation}
    s_k=\max\!\bigl(|u_y^{\mathrm{top}}(k)|,\epsilon\bigr),
    \qquad
    \widetilde{u}_{\alpha}(k)
    =
    \operatorname{clip}\!\left(\frac{u_\alpha(k)}{s_k},-1,1\right),
    \qquad \alpha\in\{x,y\},
    \label{eq:disp_norm}
\end{equation}
where $\epsilon>0$ prevents division by zero and $u_y^{\mathrm{top}}(k)=t_k\,u_r$ is the applied Dirichlet value on $\Gamma_{\mathrm{top}}$.
The normalization scale is determined from the prescribed loading condition and is therefore available during both training and deployment.
The sparse displacement-magnitude channel is similarly rescaled by a global factor computed from the training set only.

\paragraph{Network architecture.}
The damage and displacement fields are predicted by three independent but identically structured CNN2D--ConvGRU networks, one each for $\phi$, $u_x$, and $u_y$ (Fig.~\ref{fig:model_architecture}).
Each network is fully convolutional and processes the frames in $\mathcal{T}$ one at a time.
A convolutional encoder with a stride-two layer, group normalization, and ReLU activations maps every frame to a spatial feature map.
That map is then propagated through a convolutional gated recurrent unit (ConvGRU) \cite{ballas2016,shi2015}.
For an input feature map $\mathbf{x}_t$ and hidden state $\mathbf{h}_{t-1}$, the ConvGRU cell computes
\begin{align}
    \mathbf{z}_t &= \sigma\big(W_z \ast [\mathbf{x}_t,\mathbf{h}_{t-1}]\big), \\
    \mathbf{r}_t &= \sigma\big(W_r \ast [\mathbf{x}_t,\mathbf{h}_{t-1}]\big), \\
    \tilde{\mathbf{h}}_t &= \tanh\big(W_h \ast [\mathbf{x}_t,\mathbf{r}_t\odot\mathbf{h}_{t-1}]\big), \\
    \mathbf{h}_t &= (1-\mathbf{z}_t)\odot\mathbf{h}_{t-1} + \mathbf{z}_t\odot\tilde{\mathbf{h}}_t,
    \label{eq:convgru}
\end{align}
where $\ast$ denotes $3\times3$ 2D convolution in Eq.~\eqref{eq:convgru}, $\odot$ the Hadamard product, $\sigma$ the logistic sigmoid, and $\mathbf{z}_t$, $\mathbf{r}_t$ are the update and reset gates.
The encoder is the three-layer stack $\mathrm{Conv}(C_{\mathrm{in}}\!\to\!32)\to\mathrm{Conv}_{s2}(32\!\to\!C_h)\to\mathrm{Conv}(C_h\!\to\!C_h)$.
A single stride-two convolution halves the spatial resolution.
Every convolution is followed by group normalization (8 groups) and a ReLU activation.
The recurrent state has $C_h=32$ hidden channels in a $3\times3$ ConvGRU cell.
After the final frame of $\mathcal{T}$, the decoder $\mathrm{ConvT}_{s2}(C_h\!\to\!32)\to\mathrm{Conv}(32\!\to\!16)\to\mathrm{Conv}(16\!\to\!1)$ returns a single-channel field at the original resolution.
The phase-field network takes two input channels (sparse displacement magnitude and $\phi'$).
Each displacement network takes three channels ($u_x$, $u_y$, and $\phi'$).
All three networks emit one output channel.
All learned layers are convolutional, so the same weights can be evaluated at different raster sizes.
This does not by itself guarantee resolution consistency: kernel meaning and receptive field change with pixel spacing.
Transfer from $256\times256$ to $512\times512$ is therefore an empirical test (Section~\ref{sec:scaling}).
Table~\ref{tab:ml_arch} summarizes the layout.

\begin{table}[htbp]
    \centering
    \caption{CNN2D--ConvGRU surrogate architecture and input/output layout. The same fully convolutional weights can be evaluated at different raster sizes.}
    \label{tab:ml_arch}
    \begin{tabular}{ll}
        \toprule
        Component & Setting \\
        \midrule
        Networks & 3 independent CNN2D--ConvGRU (for $\phi$, $u_x$, $u_y$) \\
        Past-history input $\mathcal{T}$ & $10$ steps ($+$ current sparse step) \\
        Sensor grid & $32\times32$ (bilinear down/up-sample) \\
        Native training resolution & $256\times256$ \\
        $\phi$-network input / output channels & $2$ ($|u|_{\mathrm{sparse}}$, $\phi'$) / $1$ \\
        $u_x,u_y$-network input / output channels & $3$ ($u_x$, $u_y$, $\phi'$) / $1$ \\
        Encoder & $\mathrm{Conv}(C_{\mathrm{in}}{\to}32)\,{\to}\,\mathrm{Conv}_{s2}(32{\to}32)\,{\to}\,\mathrm{Conv}(32{\to}32)$ \\
        Hidden channels $C_h$ & 32 \\
        ConvGRU kernel & $3\times3$ \\
        Decoder & $\mathrm{ConvT}_{s2}(32{\to}32)\,{\to}\,\mathrm{Conv}(32{\to}16)\,{\to}\,\mathrm{Conv}(16{\to}1)$ \\
        Normalization / activation & GroupNorm (8 groups) / ReLU \\
        \bottomrule
    \end{tabular}
\end{table}

The CNN2D--ConvGRU architecture is selected because it preserves spatial structure while propagating temporal information through the loading history.
Table~\ref{tab:ml_arch} fixes the practical layout used throughout the paper: three identical heads, a $10$-step history, $32\times32$ sparse sensing on the $256\times256$ training raster, and a $C_h=32$ ConvGRU state with a single stride-two encoder/decoder pair.
The $\phi$ head takes two channels and each displacement head takes three, so the sensor magnitude and the shifted phase $\phi'$ are always available to the damage network, while the displacement networks also see both displacement components.
Because every layer is convolutional, the same weights can be evaluated at $512\times512$ without architectural change; that transfer is tested empirically in Section~\ref{sec:scaling}.
Its suitability relative to alternative spatial and temporal backbones is evaluated in Section~\ref{sec:architecture_benchmark}.

\paragraph{Fracture-aware loss.}
Spatial masks derived from the phase field emphasize intact material and the damage interface when scoring displacements. With threshold $\tau=1.5$ (the midpoint of the shifted phase range), softness $\beta$, and interface band $b$, we define a soft solid mask and an interface mask,
\begin{equation}
    s(\phi') = \sigma\!\left(\frac{\tau-\phi'}{\beta}\right),\qquad
    w_{\mathrm{int}}(\phi') = \exp\!\left(-\frac{|\phi'-\tau|}{b}\right).
    \label{eq:masks}
\end{equation}
Displacement accuracy is evaluated primarily in intact material and near the damage interface. Fully damaged regions are downweighted because the diffuse phase-field representation does not assign a physically meaningful zero-displacement condition inside the crack band.
The phase-field network is trained with a plain mean-squared error on $\phi'$,
\begin{equation}
    \mathcal{L}_\phi = \big\lVert \hat{\phi}' - \phi' \big\rVert_2^2 ,
    \label{eq:loss_phi}
\end{equation}
while the displacement networks minimize the mask-weighted objective
\begin{equation}
    \mathcal{L}_u
    =
    \frac{\sum (s+\gamma_{\mathrm{int}}w_{\mathrm{int}})
    |\widehat{u}-u|^2}
    {\sum (s+\gamma_{\mathrm{int}}w_{\mathrm{int}})},
    \label{eq:loss_u}
\end{equation}
where $\gamma_{\mathrm{int}}$ is the interface weight.
Mask weights, normalization constants, and data-split settings are collected in Appendix~\ref{sec:app_ml_train} (Table~\ref{tab:ml_train}).

\paragraph{Optimization and deployment.}
The two branches are optimized independently with separate Adam optimizers at learning rate $10^{-3}$, following a staged $\phi$-then-displacement schedule ($30$/$60$/$10$ epochs per stage, $100$ epochs total).
At the end of each epoch the model is evaluated on the report set, and the checkpoint minimizing the combined report-set error from Eq.~\eqref{eq:loss_phi} and Eq.~\eqref{eq:loss_u}, $\mathcal{L}_\phi+\mathcal{L}_u$, is retained.
At deployment, the model is applied sequentially with sparse measurement assimilation at every load step, starting from an initial history window $\mathcal{T}$.
Predicted fields are compared against the finite element reference both pointwise and through the mean phase-field damage diagnostic of Eq.~\eqref{eq:mean_phi}.

\begin{figure}[htbp]
    \centering
    \includegraphics[width=\linewidth]{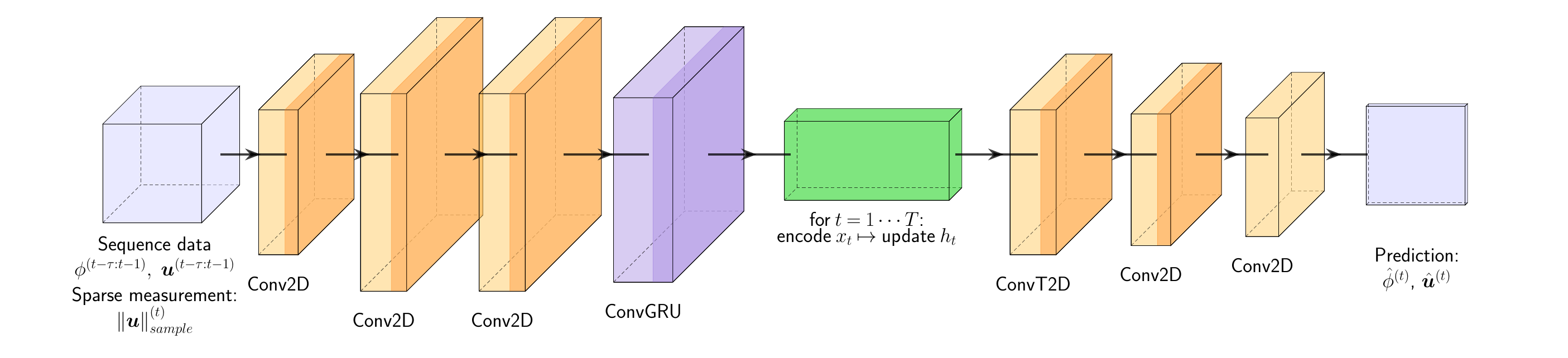}
    \caption{Architecture of the CNN2D--ConvGRU model. Two-dimensional convolutions encode local spatial patterns in the multimodal field stack at each time step; convolutional GRU states propagate temporal information from the history window $\mathcal{T}$ so that the current full-field state is reconstructed from past fields and sparse current-step measurements.}
    \label{fig:model_architecture}
\end{figure}

Figure~\ref{fig:model_architecture} summarizes how Table~\ref{tab:ml_arch} is realized in the forward pass.
Each frame of $\mathcal{T}$ is encoded to a spatial feature map, the ConvGRU updates a hidden state with convolutional gates as in Eq.~\eqref{eq:convgru}, and a transpose-convolution decoder returns a single-channel field at the original raster size.
Running three independent copies for $\phi$, $u_x$, and $u_y$ keeps the losses of Section~\ref{sec:methods} separable during the staged schedule, while sharing the same convolutional inductive bias across fields.
The schematic also clarifies why spatial context is retained throughout the loading history.
The recurrent state is a feature map rather than a pooled vector.
That property is tested against flat GRU, LSTM, RNN, and neural ODE alternatives in Section~\ref{sec:architecture_benchmark}.

\section{Results}\label{sec:results}

This section evaluates the CNN2D--ConvGRU surrogate introduced in Section~\ref{sec:methods} against the FEniCS phase-field finite element trajectories governed by Eq.~\eqref{eq:weak_u}--Eq.~\eqref{eq:weak_phi}.
Unless stated otherwise, the reference solution is the rasterized output of the staggered solver with history field Eq.~\eqref{eq:history}, volumetric--deviatoric split Eq.~\eqref{eq:psi_split}, and parameters in Appendix Table~\ref{tab:pf_sim}.
All reported diagnostics in this section (parity, field comparisons, mean-damage trajectories, and timing) use the split of Appendix Table~\ref{tab:ml_train}: $75$ training and $25$ report trajectories, with the sparse-displacement protocol and past-history input $\mathcal{T}$ of Section~\ref{sec:methods} and Fig.~\ref{fig:input_output}.
Architecture and optimizer comparisons use the same partition (Appendix~\ref{sec:app_benchmark}).

\subsection{Training results}

We optimize the three independent CNN2D--ConvGRU networks (Table~\ref{tab:ml_arch}) with the staged schedule in Appendix Table~\ref{tab:ml_train}.
The schedule uses $30$ epochs on the phase-field branch minimizing Eq.~\eqref{eq:loss_phi}, $60$ epochs on the displacement branches with mask-weighted Eq.~\eqref{eq:loss_u}, and $10$ epochs of joint fine-tuning.
The checkpoint retained for all subsequent deployments minimizes the combined report-set objective $\mathcal{L}_\phi+\mathcal{L}_u$.
Displacement learning uses the intact/interface weights in Eq.~\eqref{eq:masks}--\eqref{eq:loss_u}.
No zero-displacement projection is applied inside the diffuse crack band.

\begin{figure}[htbp]
    \centering
    \includegraphics[width=\linewidth]{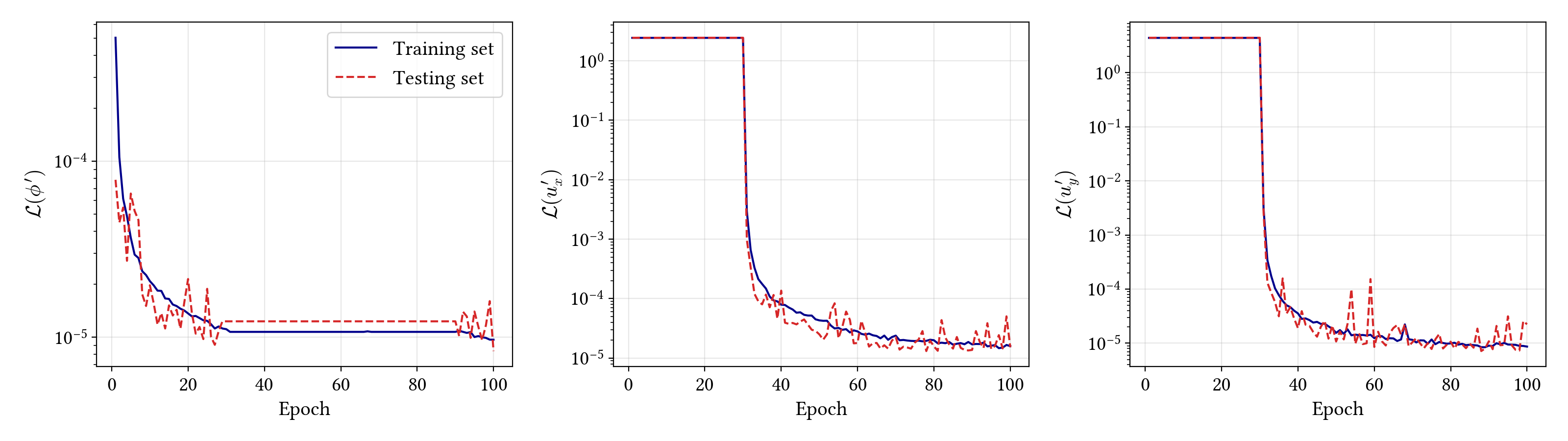}
    \caption{Training and report-set loss versus epoch for the CNN2D--ConvGRU model. Convergence of the report-set curve indicates that the learned spatiotemporal map generalizes beyond the optimization set used for parameter updates.}
    \label{fig:loss}
\end{figure}

Figure~\ref{fig:loss} summarizes the training and report-set loss histories for the $\phi$, $u_x$, and $u_y$ networks under Adam.
Both branches decrease over the $100$ epochs, and the report-set curves follow the training trends without a pronounced gap, indicating that the learned spatiotemporal map extrapolates to unseen crack layouts rather than memorizing simulation-specific details.
The staged $\phi$-then-displacement schedule appears as sharp changes at the stage boundaries.
$\mathcal{L}(\phi')$ falls rapidly during the first $30$ epochs while only the phase-field branch is updated.
It then stays essentially flat from epochs $31$--$90$ because the $\phi$ weights are frozen while the displacement networks train.
A small upward step near epoch $91$ marks joint fine-tuning, which re-enables $\phi$ updates at learning rate $10^{-3}$.
These plateaus follow from the curriculum freeze, not from a delayed ``grokking'' transition.
The logged $\phi$ loss stops changing when $\phi$ is not optimized.
Displacement losses stay at their uninitialized level until the displacement stage unlocks at epoch $31$.
Throughout training we retain the checkpoint that minimizes the combined report-set objective $\mathcal{L}_\phi+\mathcal{L}_u$, so the deployed weights need not coincide with the final epoch.

\paragraph{Optimizer comparison.}
To select the optimizer in Table~\ref{tab:ml_train}, we compared Adam, Muon, and SGD under the production protocol of Appendix~\ref{sec:app_benchmark}.
We keep train/report $75$/$25$, batch size $2$, FNO displacement refiner off, $\lambda_{\mathrm{phys}}=0$, learning rate $10^{-3}$, and the staged schedule $30$/$60$/$10$, changing only the optimizer.
Selection uses report-set metrics on the held-out trajectories.

\begin{table}[htbp]
\centering
\caption{Optimizer comparison for the CNN2D--ConvGRU baseline (train/report $75$/$25$; batch size $2$; FNO refiner off; $\lambda_{\mathrm{phys}}=0$; $100$ epochs; $30$/$60$/$10$ schedule; learning rate $10^{-3}$; NVIDIA L40S).
Selection uses report-set losses; mean $\phi$ $R^2$ is reported after selection.
``Stable'': finite report-set displacement losses (\benchpass{}: finite; \benchfail{}: diverged).}
\label{tab:optimizer_sweep}
\begin{tabular}{lcccccc}
\toprule
Optimizer & $\phi$ $R^2$ & Rep.\ $\mathcal{L}(\phi')$ & Rep.\ $\mathcal{L}(u_x)$ & Rep.\ $\mathcal{L}(u_y)$ & Time [h] & Stable \\
\midrule
Adam & $0.9997$ & $8.35\times10^{-6}$ & $1.52\times10^{-5}$ & $2.24\times10^{-5}$ & $66.1$ & \benchpass \\
Muon & $0.9997$ & $8.52\times10^{-6}$ & $1.60\times10^{-5}$ & $7.35\times10^{-6}$ & $65.2$ & \benchpass \\
SGD & $0.9967$ & $1.08\times10^{-4}$ & $8.63\times10^{-5}$ & $2.79\times10^{-4}$ & $63.9$ & \benchpass \\
\bottomrule
\end{tabular}
\end{table}

Table~\ref{tab:optimizer_sweep} shows that Adam and Muon keep finite report-set displacement losses, reach final report-set losses near $10^{-5}$, and finish in essentially matched wall-clock time on an NVIDIA L40S ($\approx65$--$66\,\mathrm{h}$).
Muon attains a slightly lower best combined report-set selection metric than Adam, while Adam remains the production default used for the architecture suites and reported deployment.
SGD also remains finite after the displacement stage unlocks at epoch $31$.
Its report-set $\mathcal{L}(u_x)$ and $\mathcal{L}(u_y)$ stay about one to two orders of magnitude larger than Adam and Muon (final values near $10^{-4}$ versus $10^{-5}$).
We therefore reject SGD on report-set displacement quality.
We retain Adam for all subsequent experiments and treat Muon as a competitive alternative for the mask-weighted displacement objective.

\begin{figure}[htbp]
    \centering
    \includegraphics[width=\linewidth]{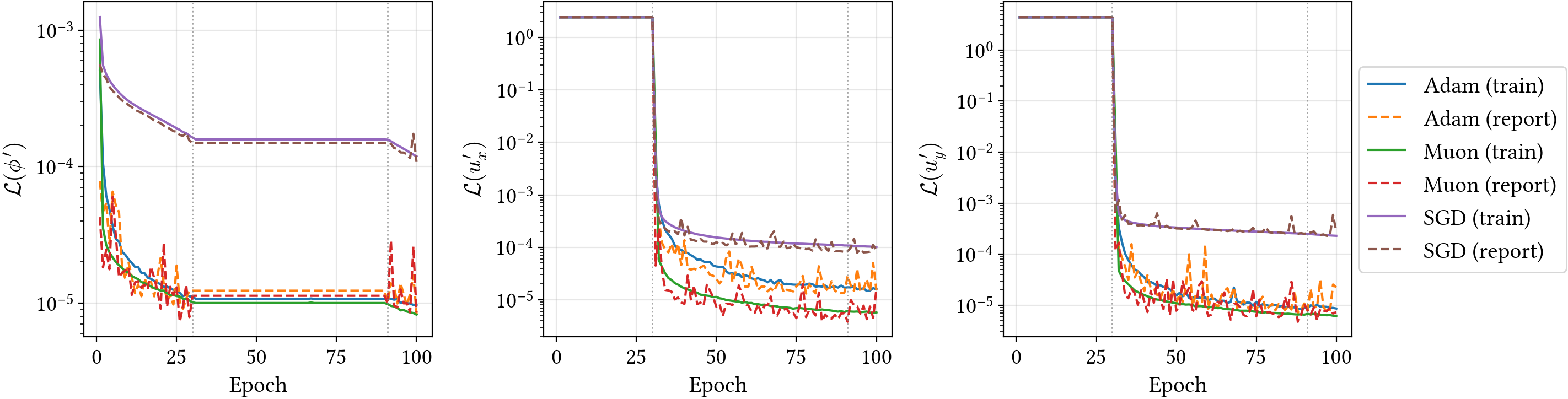}
    \caption{Training (solid) and report-set (dashed) losses for Adam, Muon, and SGD under the fixed $30$/$60$/$10$ staged schedule.
    Vertical dotted lines mark the end of the $\phi$-only stage (epoch $30$) and the start of joint fine-tuning (epoch $91$).
    Displacement losses remain high until the displacement branch is unlocked; Adam and Muon then drop by roughly two orders of magnitude and track each other closely, while SGD improves but remains about one to two orders higher.}
    \label{fig:optimizer_sweep}
\end{figure}

Figure~\ref{fig:optimizer_sweep} makes the stage structure of Table~\ref{tab:optimizer_sweep} visible in the loss histories.
Displacement losses stay high until the displacement branch unlocks at epoch $31$; Adam and Muon then drop by roughly two orders of magnitude and track each other closely on both training and report-set curves.
SGD improves after the same unlock but remains about one to two orders higher through the displacement and fine-tuning stages.
The vertical markers at epochs $30$ and $91$ coincide with the curriculum boundaries used for all three optimizers, so the gap after epoch $31$ is attributable to optimizer choice rather than to a change in schedule.

\begin{figure}[htbp]
    \centering
        \includegraphics[width=\linewidth]{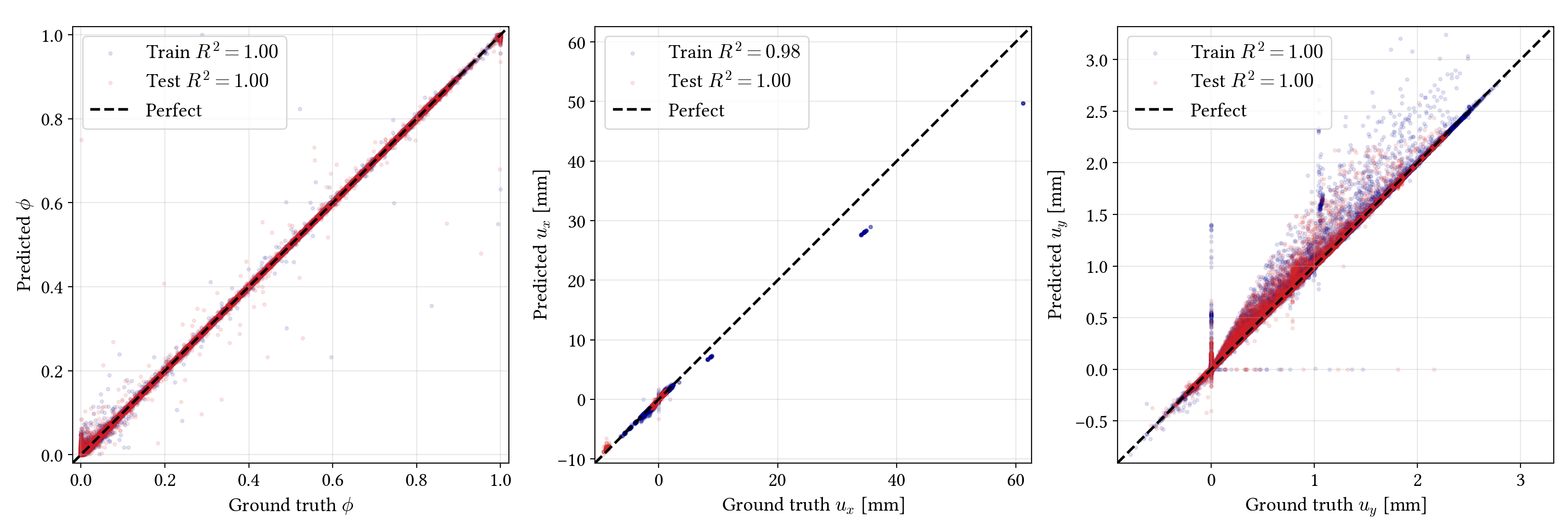}
    \caption{Parity plot comparing surrogate predictions against finite element references over the report set. Points near the diagonal indicate close agreement in the reported field or integrated metrics; systematic bias would appear as off-diagonal structure.}
    \label{fig:parity_all}
\end{figure}

Figure~\ref{fig:parity_all} compares current-step reconstructions with finite-element references across the report set for the retained Adam checkpoint.
Predictions for the shifted phase field $\phi'$ and the load-scaled displacements $\widetilde{u}_{\alpha}$ defined in Eq.~\eqref{eq:disp_norm} cluster near the identity line, with no evident systematic bias across the test trajectories.
Displacement accuracy is assessed primarily in intact material and near the damage interface; fully damaged voxels are downweighted in Eq.~\eqref{eq:loss_u} and are not forced to a zero-displacement condition.
Together, Figs.~\ref{fig:loss} and~\ref{fig:parity_all} show that the ConvGRU-based encoder--decoder in Fig.~\ref{fig:model_architecture} and Eq.~\eqref{eq:convgru} captures the dominant spatial and temporal features of the fracture data set before sequential deployment.
The model uses a history window of $\mathcal{T}=10$ steps and $C_h=32$ hidden channels, as listed in Table~\ref{tab:ml_arch}.

\subsection{Sequential deployment with sparse measurement assimilation}\label{sec:deployment}

We next apply the trained model sequentially on report-set simulations, assimilating sparse displacement measurements at every load step.
Each case is advanced on the native $256\times256$ export grid (Table~\ref{tab:pf_sim}) using the sparse $32\times32$ displacement sensing and masked current-step phase input described in Section~\ref{sec:methods}.
At each load step the model receives $\mathcal{T}$ and the current sparse measurements, reconstructs $\hat\phi'$, $\hat u_x$, and $\hat u_y$, and denormalizes displacements with the prescribed-load scale $s_k$ from Eq.~\eqref{eq:disp_norm}.

\begin{figure}[ht]
    \centering
    \includegraphics[width=.65\linewidth]{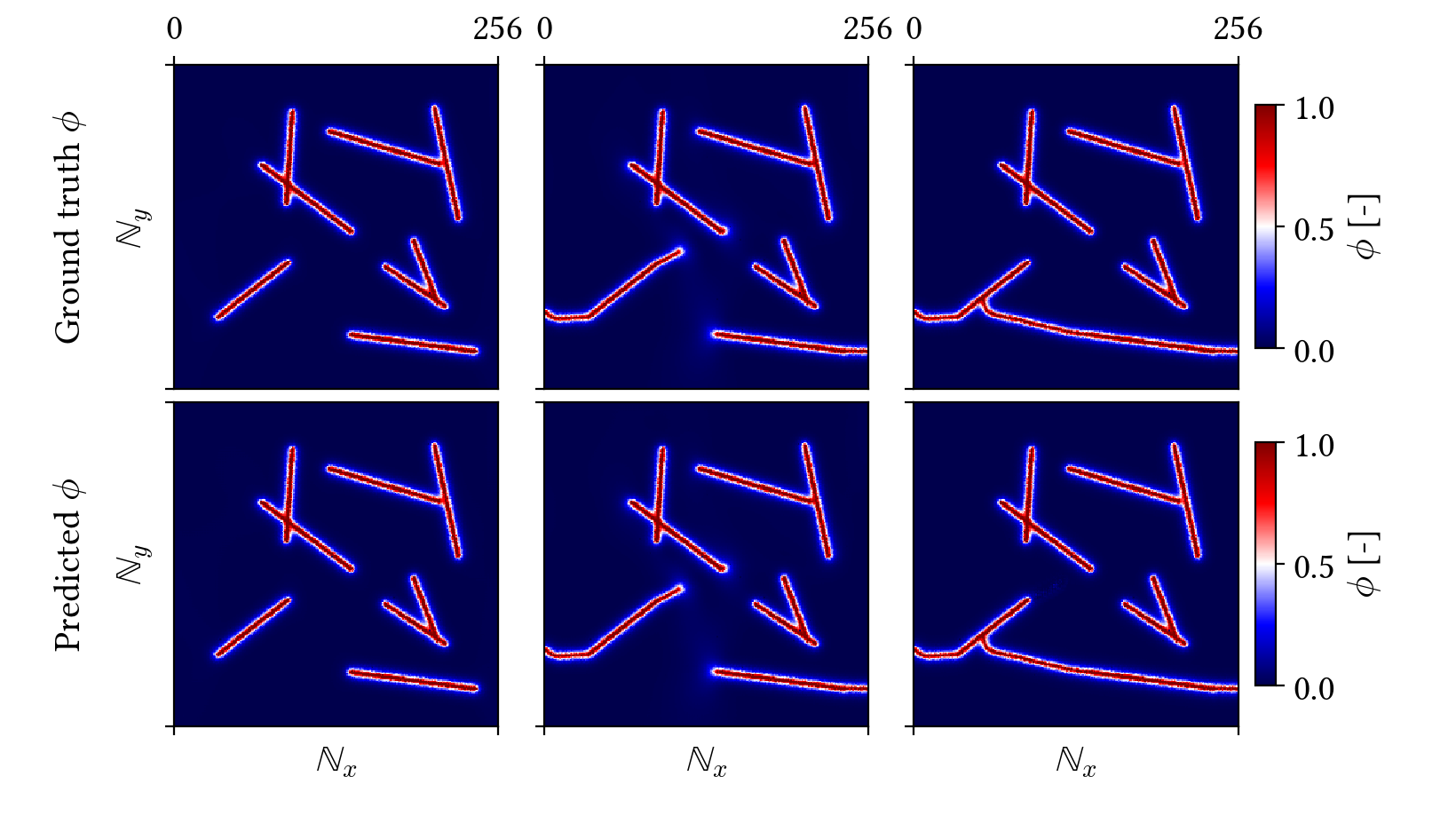}
    \caption{Phase-field damage $\phi$ for a report-set case at the training-grid resolution, shown at three increasing load steps (columns: early, intermediate, and near-terminal moments of the loading history): top row, finite element reference; bottom row, CNN2D--ConvGRU reconstruction. Crack paths and diffuse damage zones are captured with limited spurious smearing outside the fracture band. Fields are on the $1\,\mathrm{m}\times1\,\mathrm{m}$ domain, and the applied top displacement advances by $\Delta u_y\approx5\,\mu\mathrm{m}$ per converged frame along the $u_y\le2.5\,\mathrm{mm}$ ramp.}
    \label{fig:phi_pred}
\end{figure}

Figure~\ref{fig:phi_pred} compares the reference phase field (top row) with the network reconstruction (bottom row) at three moments during active crack growth.
The columns span increasing applied top displacements within the $0$--$2.5\,\mathrm{mm}$ ramp.
The per-frame increment is $\Delta u_y=\Delta t\,u_r\approx5\,\mu\mathrm{m}$.
The model reproduces the main crack branches and the diffuse bands whose width is set by $\ell$ in Table~\ref{tab:pf_sim}.
By the near-terminal column the segments have coalesced into a spanning damage band across the specimen.
Spurious damage outside the fracture corridors is limited, consistent with Eq.~\eqref{eq:loss_phi} and with irreversibility from Eq.~\eqref{eq:history} in the training data.

\begin{figure}[ht]
    \centering
    \begin{subfigure}{.5\linewidth}
        \centering
        \includegraphics[width=\linewidth]{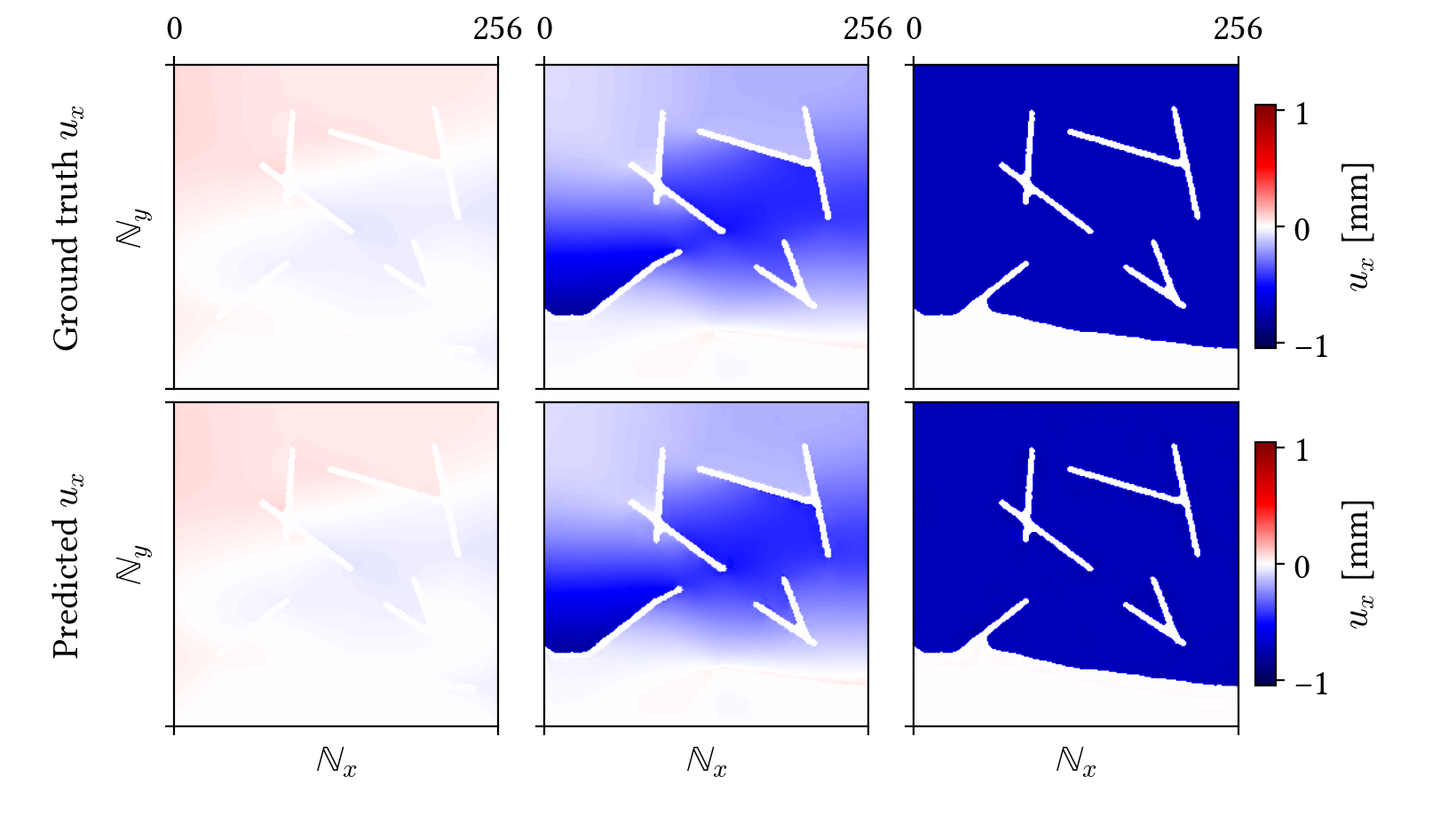}
        \label{fig:ux_pred}\caption{}
    \end{subfigure}\hfill
    \begin{subfigure}{.5\linewidth}
        \centering
        \includegraphics[width=\linewidth]{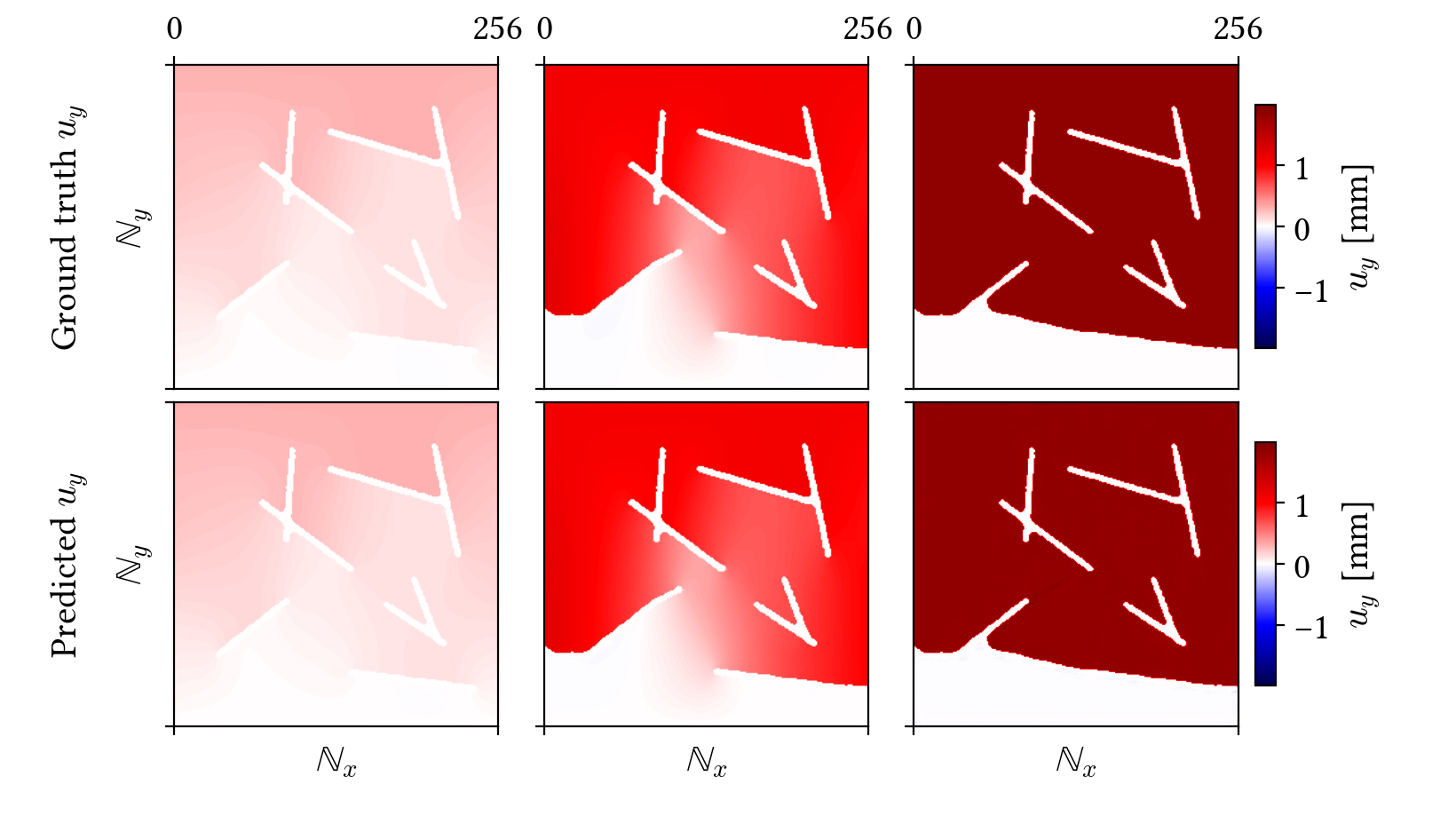}
        \label{fig:uy_pred}\caption{}
    \end{subfigure}
    \caption{Displacement fields at $256\times256$: (a) $u_x$; (b) $u_y$.
    Top: FEM; bottom: network at three load steps.
    Largest errors occur near tips and softened zones ahead of the damage front.}
    \label{fig:disp_pred}
\end{figure}

Figure~\ref{fig:disp_pred}(a) and (b) show the corresponding horizontal and vertical displacement fields at the same three load steps (top row reference, bottom row reconstruction).
The largest pointwise errors concentrate near propagating crack tips and along steep displacement gradients in the softened zones ahead of the damage front.
There the crack-driving energy in Eq.~\eqref{eq:psi_split} couples strongly to the evolving degradation $g(\phi)$ in Eq.~\eqref{eq:weak_u}.
Away from fully damaged regions, the predicted displacements match the reference bulk elastic pattern, reflecting the intact-region weighting and interface emphasis in Eq.~\eqref{eq:loss_u}.
As in Section~\ref{sec:methods}, fully damaged voxels are downweighted in Eq.~\eqref{eq:loss_u}.

We summarize the global progression of fracture through the mean phase-field damage
\begin{equation}
\bar{\phi}(t)=\frac{1}{N}\sum_{i=1}^{N}\phi_i(t),
\label{eq:mean_phi}
\end{equation}
which averages the diffuse phase field over the $N$ raster pixels.
Raw (uncorrected) trajectories of $\bar{\phi}$ are reported in Figs.~\ref{fig:damage_ratio} and~\ref{fig:damage_ratio_512}; no post hoc bias alignment is applied.

\begin{figure}[ht]
    \centering
    \includegraphics[width=.5\linewidth]{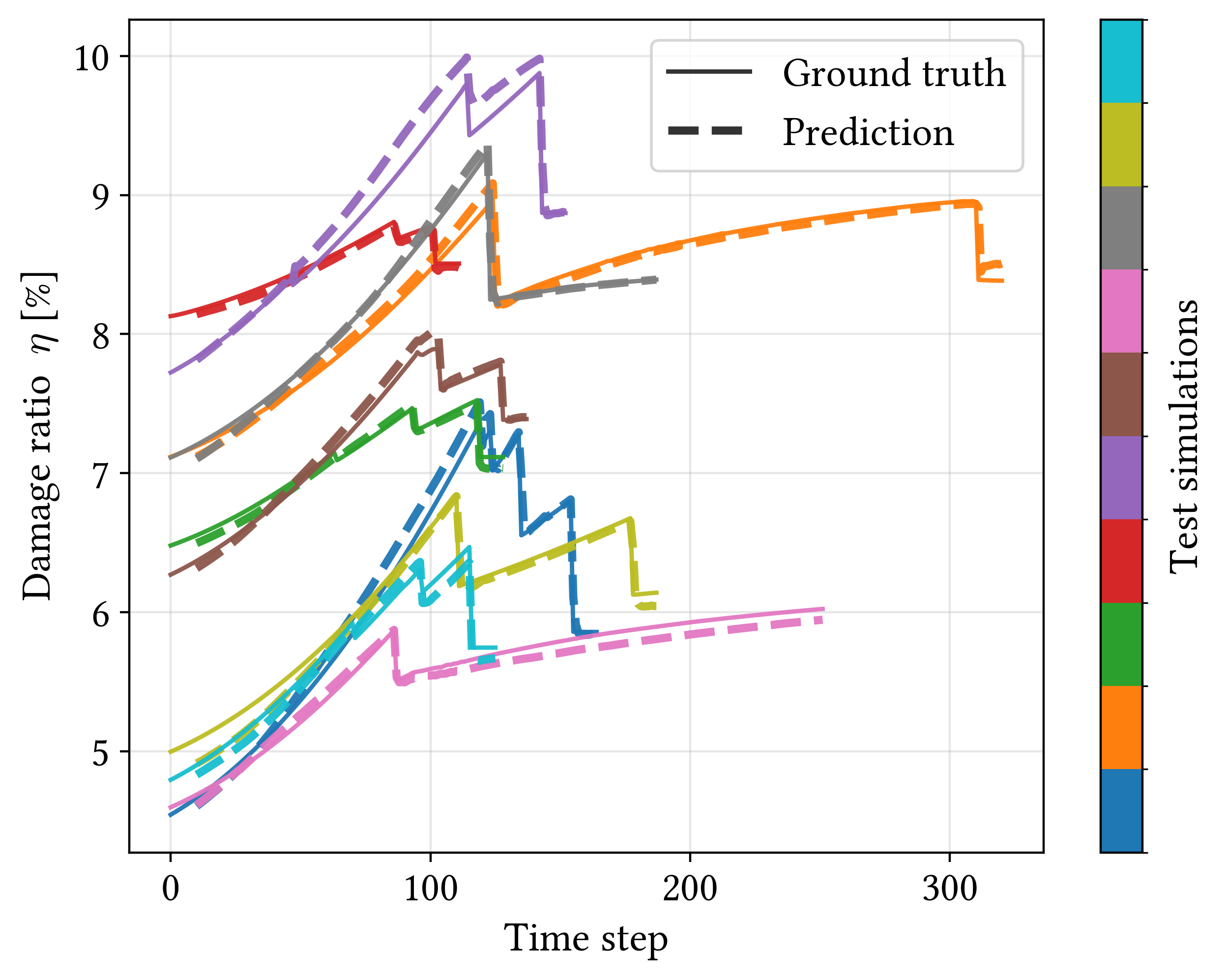}
    \caption{Mean phase-field damage $\bar{\phi}$ (Eq.~\eqref{eq:mean_phi}) versus load step for ten report-set cases at $256\times256$.
    Dashed: network; solid: FEM reference.
    Each step advances $u_y$ by $\Delta u_y\approx5\,\mu\mathrm{m}$ toward $2.5\,\mathrm{mm}$.}
    \label{fig:damage_ratio}
\end{figure}

Figure~\ref{fig:damage_ratio} reports $\bar{\phi}$ from Eq.~\eqref{eq:mean_phi} versus pseudo-time for ten report-set configurations at the native $256\times256$ grid.
Figs.~\ref{fig:phi_pred} and~\ref{fig:disp_pred} show three instants along one representative trajectory from this set.
Each load step advances the applied top displacement by $\Delta u_y\approx5\,\mu\mathrm{m}$ toward the terminal $u_y\approx2.5\,\mathrm{mm}$.
The reconstructed trajectories track the finite element references through nucleation, growth, and coalescence.
Residual late-time drift is consistent with the history-dependent nature of the reference process Eq.~\eqref{eq:history} and with compounding reconstruction errors when the model is applied sequentially over many load steps.

\subsection{Transfer to a finer raster resolution}\label{sec:scaling}

As noted in Section~\ref{sec:methods} and Table~\ref{tab:ml_arch}, the convolutional weights can be evaluated at a different raster size.
Transfer from $256\times256$ to $512\times512$ remains an empirical test: receptive field and pixel spacing change together.
The underlying $300\times300$ finite-element mesh and physical model are unchanged.
Only the neural-network raster is refined.
For the $512\times512$ evaluation, sparse displacement sensing uses a $64\times64$ grid instead of $32\times32$.
The sensor-to-raster ratio is preserved (one sample per $8\times8$ block).
On the fixed physical domain this also halves the physical sensor spacing and quadruples the number of sensors.
The experiment therefore tests joint transfer to a finer raster and a proportionally refined measurement grid.
It does not isolate raster refinement at fixed physical sensor locations.

\begin{figure}[ht]
    \centering
    \includegraphics[width=.65\linewidth]{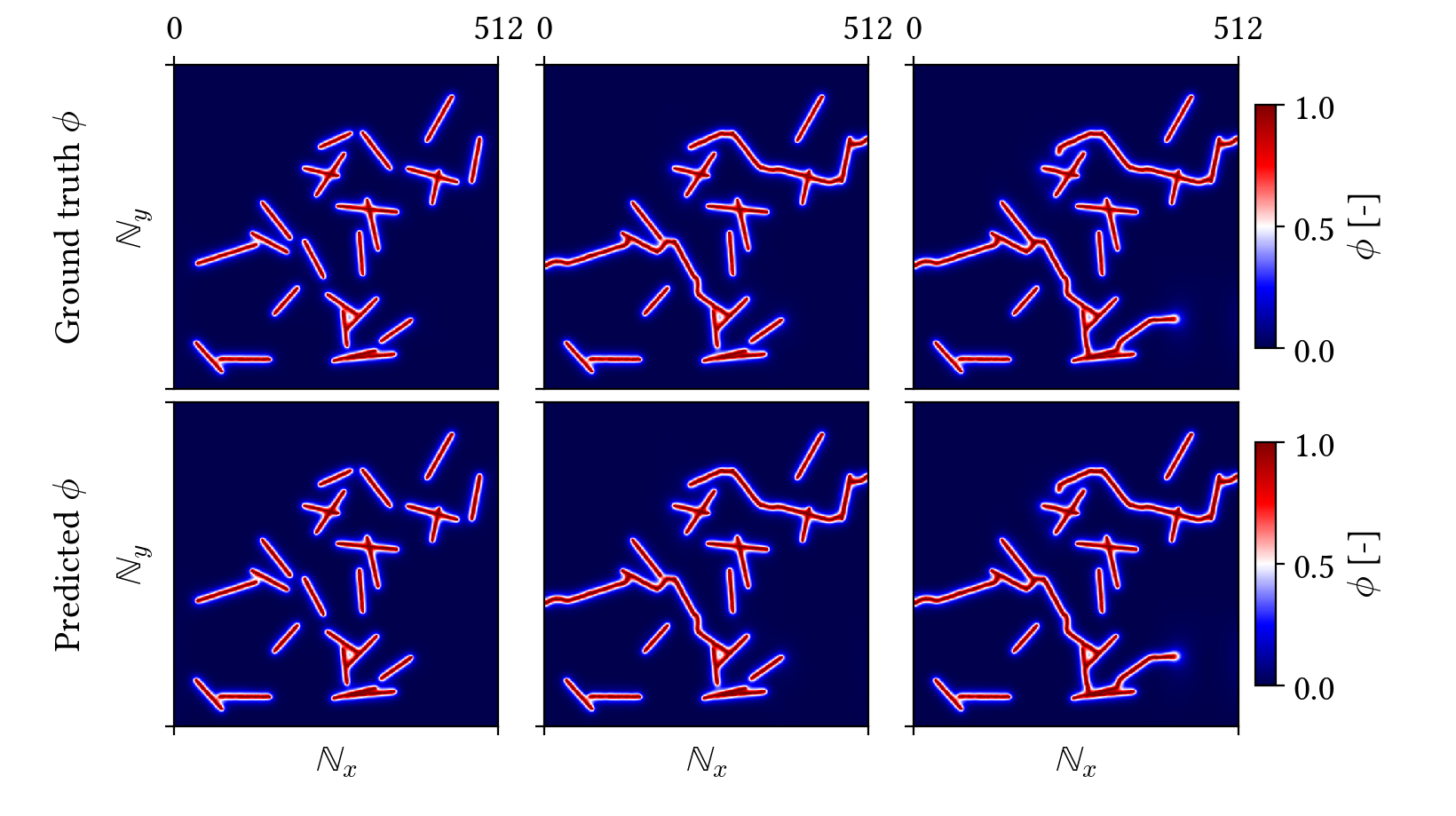}
    \caption{Phase-field $\phi$ on a $512\times512$ raster of the same physical domain (three load steps; top: FEM; bottom: sequential reconstruction without retraining).
    Domain and $300\times300$ FEM mesh are unchanged; pixel size $\approx3.9\to1.95\,\mathrm{mm}$.}
    \label{fig:phi_pred_512}
\end{figure}

Figure~\ref{fig:phi_pred_512} compares the reference phase field (top row) with the sequential network reconstruction (bottom row) on the finer $512\times512$ raster at three increasing load steps, without retraining.
The damage topology remains aligned with the reference through early, intermediate, and near-terminal loading.
Errors near tips and along interacting crack segments are slightly amplified relative to the $256\times256$ deployment in Section~\ref{sec:deployment}, as expected from error accumulation over many sequential steps on a longer spatial index set.

\begin{figure}[ht]
    \centering
    \begin{subfigure}{.5\linewidth}
        \centering
        \includegraphics[width=\linewidth]{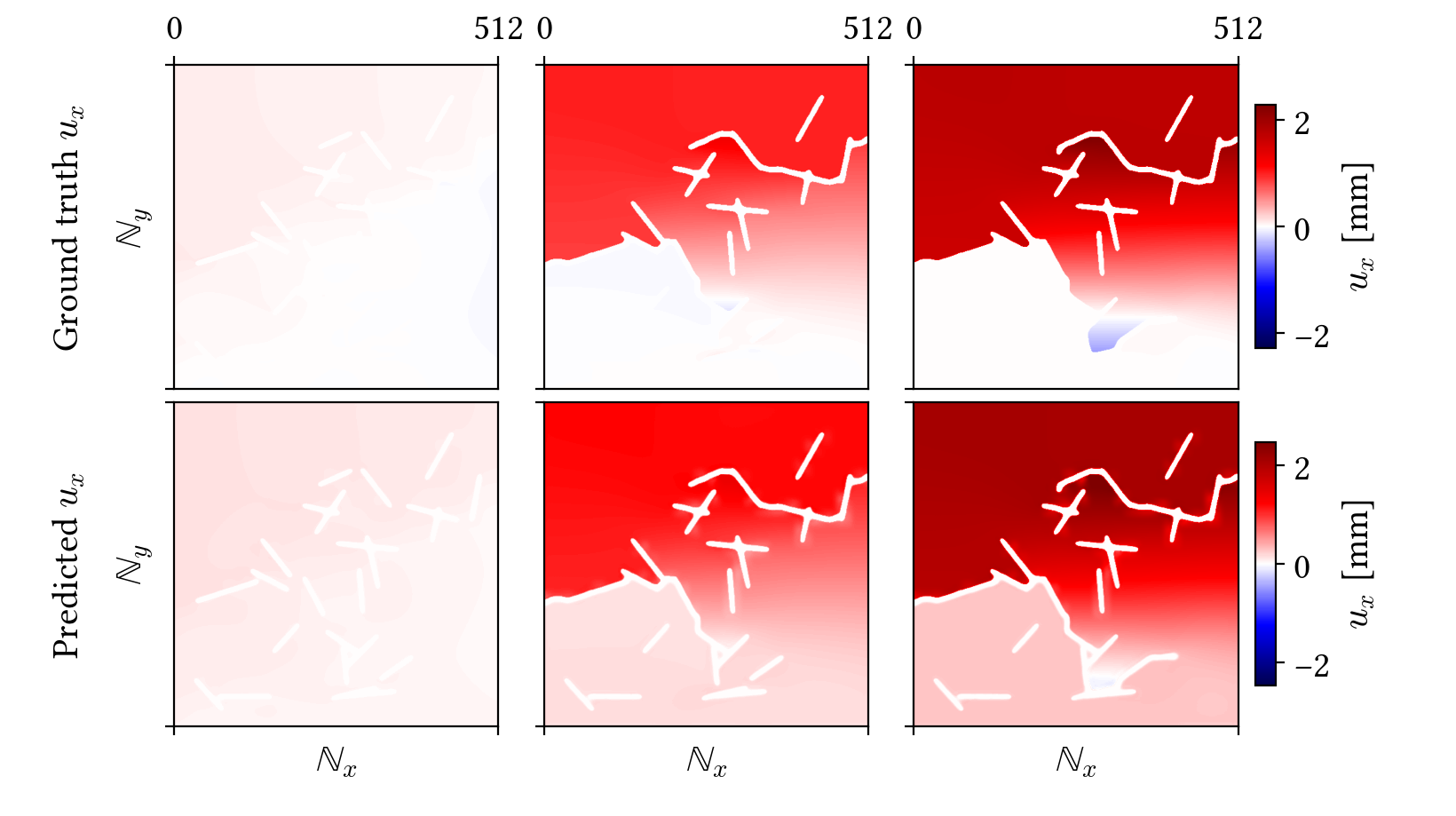}
        \label{fig:ux_pred_512}\caption{}
    \end{subfigure}\hfill
    \begin{subfigure}{.5\linewidth}
        \centering
        \includegraphics[width=\linewidth]{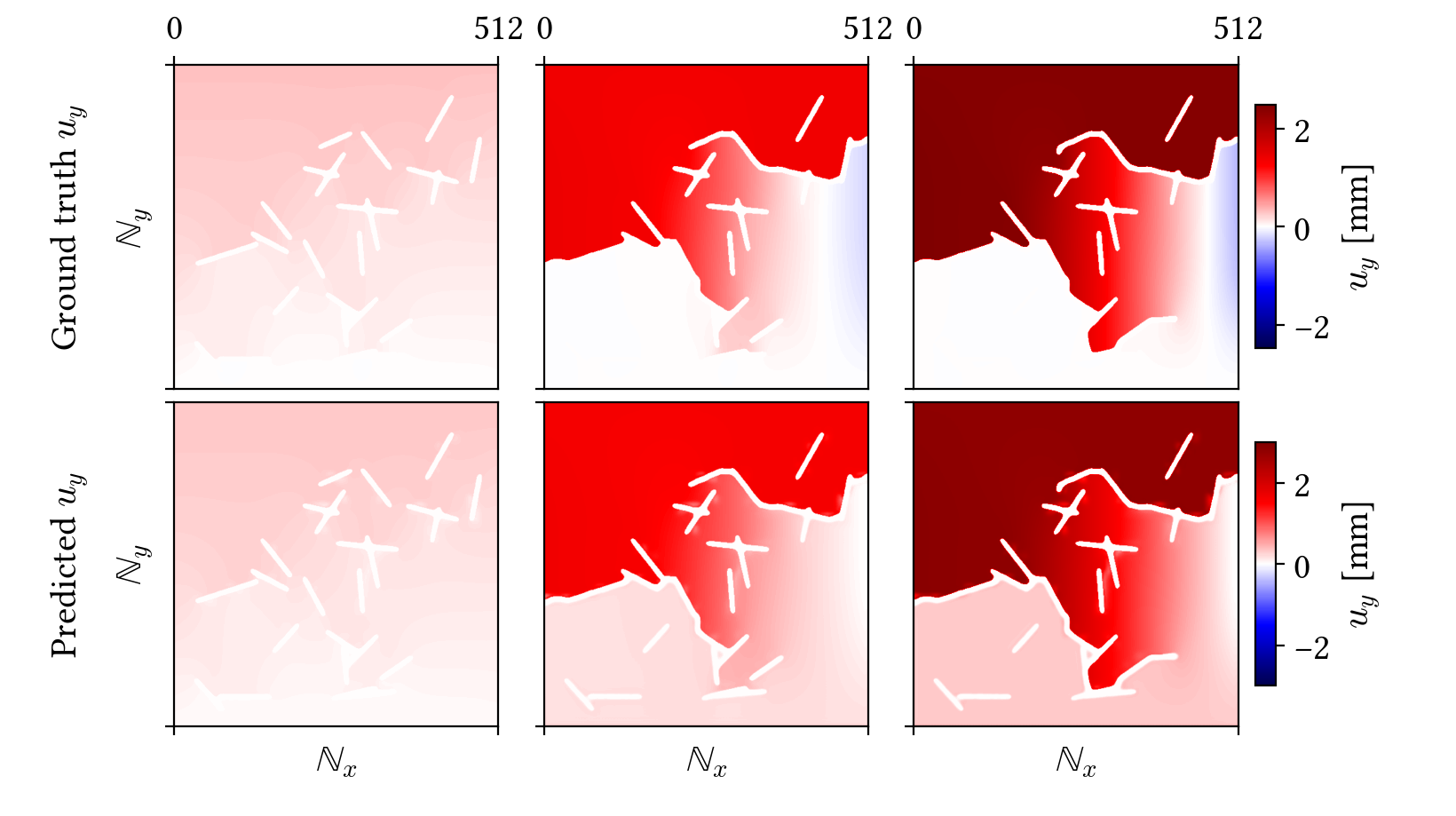}
        \label{fig:uy_pred_512}\caption{}
    \end{subfigure}
    \caption{Displacement fields on the finer $512\times512$ raster of the same physical domain: (a) $u_x$ and (b) $u_y$; in each the top row is the finite element reference and the bottom row the network reconstruction, shown at three increasing load steps (early, intermediate, and near-terminal moments). Physical size is fixed at $1\,\mathrm{m}\times1\,\mathrm{m}$ (raster pixel $\approx1.95\,\mathrm{mm}$ at $512^2$); displacements are in mm.}
    \label{fig:disp_pred_512}
\end{figure}

Figure~\ref{fig:disp_pred_512}(a) and (b) show the corresponding horizontal and vertical displacements on the same finer raster and load steps.
As on the training grid, the largest discrepancies concentrate near propagating tips and steep gradients in softened zones ahead of the damage front.
Away from fully damaged bands, the bulk elastic pattern remains consistent with the reference, indicating that the convolutional weights transfer to the refined pixel spacing when the sensor-to-raster ratio is preserved.

\begin{figure}[ht]
    \centering
    \includegraphics[width=.5\linewidth]{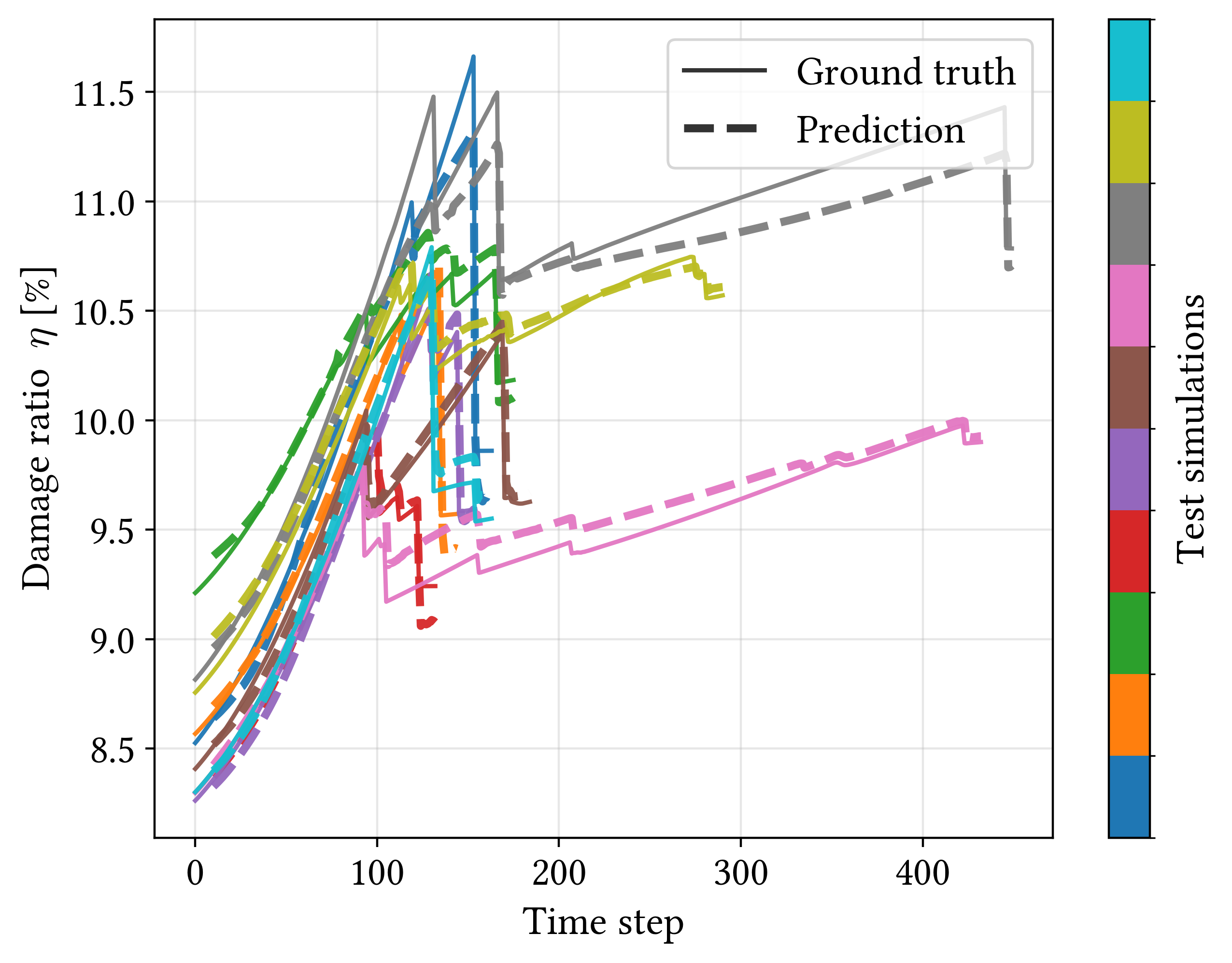}
    \caption{Mean phase-field damage $\bar{\phi}$ (Eq.~\eqref{eq:mean_phi}) versus load step for ten cases at $512\times512$.
    Dashed: network without retraining; solid: FEM reference.}
    \label{fig:damage_ratio_512}
\end{figure}

Figure~\ref{fig:damage_ratio_512} reports $\bar{\phi}$ from Eq.~\eqref{eq:mean_phi} for ten cases on the finer $512\times512$ raster of Figs.~\ref{fig:phi_pred_512} and~\ref{fig:disp_pred_512}.
The reconstructed trajectories again follow the FEM references through nucleation, growth, and coalescence.
Because Eq.~\eqref{eq:mean_phi} averages diffuse $\phi$ rather than thresholding it, the curves summarize the full sequential deployment without retraining.
$\bar{\phi}$ tracks the FEM curves more closely than the pointwise displacements in Figure~\ref{fig:disp_pred_512}(a) and (b), consistent with Fig.~\ref{fig:phi_pred_512}.

In summary, the framework shows three capabilities.
It reconstructs report-set fields at the training raster and tracks mean damage over sequential deployment (Figs.~\ref{fig:phi_pred}--\ref{fig:damage_ratio}).
It assimilates sparse displacement measurements with partial phase information (Section~\ref{sec:methods}).
It also evaluates on an unseen $512\times512$ raster without retraining (Figs.~\ref{fig:phi_pred_512}--\ref{fig:damage_ratio_512}).
Together with Figs.~\ref{fig:loss} and~\ref{fig:parity_all}, these results support CNN2D--ConvGRU as a practical measurement-conditioned accelerator for phase-field fracture state inference.

\subsection{Speedup evaluations}

We compare wall-clock cost of sequential measurement-conditioned reconstruction against regenerating the same trajectory with the FEM solver.
Table~\ref{tab:rollout_speedup_256} shows ten report-set cases at the $256\times256$ resolution (the same ones as in Fig.~\ref{fig:damage_ratio}).
For each case, we take the FEM time from saved logs and measure how long the network takes on CPU and GPU.
At every load step the network reconstructs the current fields from the previous $10$ states and sparse current-step measurements, starting once a history window is available (from step $10$).
This is the same setup used in Fig.~\ref{fig:damage_ratio}.
We report speedup as FEM time divided by network inference time.

\begin{table}[htbp]
\centering
\caption{Inference speedup vs.\ finite element reference for ten $256\times256$ report-set simulations (listed in the same order as Fig.~\ref{fig:damage_ratio}). FEM wall time is recovered from archival generation timestamps; network wall time is measured for CPU and GPU inference with repeated current-step reconstructions (history length $=10$). Steps: total number of converged load increments $T$ retained per trajectory.}
\label{tab:rollout_speedup_256}
\begin{tabular}{r r r r r r r}
\toprule
Sim.\ \# & Steps & FEM [h] & Time CPU [s] & Speedup CPU & Time GPU [s] & Speedup GPU \\
\midrule
1 & 165 & 2.97 & 57.2 & 187$\times$ & 40.1 & 266$\times$ \\
2 & 321 & 4.77 & 113.6 & 151$\times$ & 79.9 & 215$\times$ \\
3 & 129 & 2.27 & 44.6 & 183$\times$ & 30.7 & 265$\times$ \\
4 & 112 & 1.81 & 37.3 & 174$\times$ & 26.3 & 247$\times$ \\
5 & 153 & 2.72 & 52.8 & 186$\times$ & 36.6 & 268$\times$ \\
6 & 138 & 2.63 & 47.1 & 201$\times$ & 32.8 & 289$\times$ \\
7 & 252 & 3.55 & 89.8 & 142$\times$ & 62.2 & 206$\times$ \\
8 & 188 & 2.98 & 66.1 & 162$\times$ & 45.9 & 233$\times$ \\
9 & 188 & 3.21 & 67.0 & 173$\times$ & 45.9 & 252$\times$ \\
10 & 126 & 2.38 & 43.9 & 196$\times$ & 29.9 & 287$\times$ \\
\midrule
Mean & -- & 2.93 & 61.9 & 175$\times$ & 43.0 & 253$\times$ \\
\bottomrule
\end{tabular}
\end{table}

Across the ten cases, the finite element reference runs require $\approx1.8$--$4.8,\mathrm{h}$, depending on simulation length ($T=112$--$321$ load steps).
In contrast, the network completes the corresponding sequential reconstructions in under two minutes on CPU and under one minute on GPU.
The mean speedups are $175\times$ on CPU and $253\times$ on GPU.
The per-case gains range from $142\times$ to $201\times$ on CPU and from $206\times$ to $289\times$ on GPU.
Each finite element increment requires a full staggered solve on the unstructured mesh, whereas each current-step reconstruction requires only a single forward pass through the convolutional--recurrent stack.
This difference leads to a two-orders-of-magnitude reduction in wall-clock cost for these quasi-static load histories.
We do not tabulate the finer-raster $512\times512$ deployment study of Section~\ref{sec:scaling} here.
That study refines only the exported raster, while the underlying FEniCS discretization is unchanged; therefore, a direct timing comparison would conflate raster-resolution transfer with inference acceleration.
The $512\times512$ results in Figs.~\ref{fig:phi_pred_512}--\ref{fig:damage_ratio_512} instead demonstrate evaluation on a finer raster representation of the same physical domain.
The present table isolates the main practical question for deployment: how much faster measurement-conditioned reconstruction is than repeated high-fidelity staggered solves at the training-grid raster.

\subsection{Benchmarking data-driven architectures}\label{sec:architecture_benchmark}

To select the architecture in Section~\ref{sec:methods} and Table~\ref{tab:ml_arch}, we ran two suites on the $256\times256$ crack ensemble in Table~\ref{tab:pf_sim} under the production protocol.
Each candidate keeps the three-network layout for $\phi$, $u_x$, and $u_y$.
Each matches the baseline parameter budget within $\pm5\%$ (FNO displacement refiner off; budget $\approx2.52\times10^{5}$).
Each follows the staged schedule and hyperparameters in Appendix Tables~\ref{tab:ml_train}--\ref{tab:benchmark_protocol} for $100$ epochs.
Every run was carried out on a single NVIDIA L40S GPU.

We split the $100$ crack-pattern simulations by trajectory into $75$ for training and $25$ for reporting in Table~\ref{tab:benchmark_combined}.
Architecture ranking uses report-set performance (pass if $\phi$ $R^2\ge0.99$) together with training wall time among passing models.
All architectures share this split, along with the same losses, sensing setup, and training schedule.

In the \emph{spatial} suite we replace the CNN2D encoder--decoder with FNO2D, DeepONet, or PointNet while keeping the ConvGRU temporal core fixed.
In the \emph{temporal} suite we keep the CNN2D stack and swap ConvGRU for a flat GRU, LSTM, vanilla RNN, or neural ODE integrator on spatially pooled features.
Appendix~\ref{sec:app_benchmark} collects the remaining protocol details.

\begin{table}[htbp]
\centering
\caption{Architecture benchmark summary ($75$ training / $25$ report trajectories; $100$ epochs; matched parameter budget without FNO refiner; NVIDIA L40S).
Mean $\phi$ $R^2$ on the report set; a model passes if $R^2_{\phi}\ge0.99$.}
\label{tab:benchmark_combined}
\begin{tabular}{lcccccc}
\toprule
 & \multicolumn{3}{c}{Spatial} & \multicolumn{3}{c}{Temporal} \\
\cmidrule(lr){2-4} \cmidrule(lr){5-7}
Backbone & $\phi$ $R^2$ & Time [h] & Pass & $\phi$ $R^2$ & Time [h] & Pass \\
\midrule
CNN2D--ConvGRU (baseline) & $0.9998$ & $66.22$ & \benchpass & $0.9997$ & $65.77$ & \benchpass \\
FNO2D--ConvGRU & $0.9996$ & $69.29$ & \benchpass & -- & -- & -- \\
DeepONet--ConvGRU & $0.0580$ & $63.84$ & \benchfail & -- & -- & -- \\
PointNet--ConvGRU & $0.9997$ & $64.13$ & \benchpass & -- & -- & -- \\
CNN2D--GRU & -- & -- & -- & $0.0035$ & $59.19$ & \benchfail \\
CNN2D--LSTM & -- & -- & -- & $0.0035$ & $59.06$ & \benchfail \\
CNN2D--RNN & -- & -- & -- & $0.0031$ & $55.80$ & \benchfail \\
CNN2D--Neural ODE & -- & -- & -- & $0.0031$ & $66.65$ & \benchfail \\
\bottomrule
\end{tabular}
\end{table}

Table~\ref{tab:benchmark_combined} reports mean $\phi$ $R^2$ on the report set, training wall time, and whether each model meets $R^2_{\phi}\ge0.99$ on that set.
In the spatial suite of Table~\ref{tab:benchmark_combined}, CNN2D--ConvGRU, FNO2D--ConvGRU, and PointNet--ConvGRU all pass with $R^2_{\phi}\ge0.9996$.
DeepONet--ConvGRU fails ($R^2_{\phi}\approx0.058$).
In the temporal suite, only CNN2D--ConvGRU passes ($R^2_{\phi}\approx0.9997$).
Flat GRU, LSTM, vanilla RNN, and neural ODE cores remain near $R^2_{\phi}\approx0.003$.
Among spatial encoders that pass, PointNet is slightly cheaper than CNN2D ($\approx64\,\mathrm{h}$ versus $\approx66\,\mathrm{h}$).
FNO2D is the most expensive ($\approx69\,\mathrm{h}$).
Relative to CNN2D--ConvGRU, FNO2D--ConvGRU does not improve mean report-set $\phi$ $R^2$ ($0.9996$ versus $0.9998$) under the matched parameter budget.
It only adds training cost.
For this sparse-sensing, locally structured crack-reconstruction task, a Fourier neural operator encoder is therefore not useful: global spectral mixing does not buy accuracy over local convolutions, so we keep the CNN2D encoder and leave the optional FNO displacement refiner off in the production protocol (Tables~\ref{tab:ml_train} and~\ref{tab:benchmark_protocol}).
The temporal failures show that pooling spatial features before a flat recurrent or ODE update is insufficient under the matched budget.
We retain CNN2D--ConvGRU as the reported baseline.
Convolutional recurrence is required by the temporal suite, it attains the highest spatial $R^2_{\phi}$, and it matches the architecture used for the deployment studies above.

\section{Discussion}\label{sec:discussions}

A central practical advantage of the proposed model is evaluation on finer or coarser rasters without architectural change \cite{lecun1998,ballas2016,shi2015,geneva2020}.
Section~\ref{sec:methods} already notes that convolutional weights transfer across raster size only as an empirical property.
Section~\ref{sec:scaling} tests that property for $256\times256\to512\times512$ with a proportionally refined sensor grid on the same physical domain and $300\times300$ FEM mesh.
Figs.~\ref{fig:phi_pred_512}--\ref{fig:damage_ratio_512} remain aligned with the FEM reference through early fracture.
We do not report downscaling experiments here.
The motivation is computational: high-fidelity phase-field finite element trajectories remain expensive to generate and march, especially when the regularization length must be resolved and irreversibility couples many nonlinear solves \cite{bourdin2007,ambati2015,miehe2010a}.
Replacing repeated staggered updates with measurement-conditioned sequential inference amortizes that cost across parameter studies, design exploration, and real-time screening, consistent with broader data-driven model reduction for mechanics \cite{lee2020,geneva2020}.

The spatial architecture suite in Section~\ref{sec:architecture_benchmark} also clarifies the role of Fourier operator layers for this problem.
FNO2D--ConvGRU meets the pass threshold in Table~\ref{tab:benchmark_combined}, but it does not outperform CNN2D--ConvGRU on report-set $\phi$ $R^2$ and it trains more slowly.
Under a matched parameter budget, global Fourier mixing therefore adds little for reconstructing localized crack bands from sparse sensors and a short field history.
We conclude that FNO is not useful for the present setting and retain local convolutions without an FNO displacement refiner.

At the finer raster, Fig.~\ref{fig:damage_ratio_512} shows that the mean phase-field damage $\bar{\phi}$ from Eq.~\eqref{eq:mean_phi}, together with the phase field $\phi$ in Fig.~\ref{fig:phi_pred_512}, are reproduced more closely than the displacement components.
Figure~\ref{fig:disp_pred_512}(a) and (b) show larger pointwise $u_x$ and $u_y$ errors than in the $256\times256$ report-set deployment in Section~\ref{sec:deployment}, even though bulk trends, crack topology, and the early part of the load path remain well captured.
This pattern is expected for three reasons.
First, sequential deployment compounds one-step reconstruction errors over many increments, a known limitation of learned time-step models \cite{geneva2020,goswami2022}.
Appendix~\ref{app:sequential_error} records a conditional discrete Gr\"onwall estimate for sequential reconstruction with shared current-step measurements.
History-window mismatches can accumulate when a local Lipschitz assumption holds on a restricted set of histories.
The estimate does not assert that the phase-field fracture evolution is globally Lipschitz, nor that the network is an autonomous time integrator.
It formalizes why late-time drift in Figs.~\ref{fig:damage_ratio} and~\ref{fig:damage_ratio_512} can grow with deployment length even when single-step errors remain small.
Second, the network was trained on $75$ simulations at $256\times256$, so fine-raster displacement detail was never part of the optimization data.
Third, the $512\times512$ test in Section~\ref{sec:scaling} evaluates simultaneous transfer to a finer raster and a proportionally refined measurement grid rather than fully unseen crack-pattern generalization; displacement fields are therefore evaluated under extrapolation in spatial sampling.
For trend-level failure screening, the results track where cracks grow and how global mean damage evolves.
They remain useful for that purpose \cite{ambati2015}.
Displacement-dominated quantities near tips should be interpreted with this raster-resolution and sequential-deployment caveat, in line with other learned brittle-fracture surrogates \cite{goswami2022}.

The sequence-learning viewpoint also clarifies how the model relates to the variable pacing of the reference solver.
The ground truth trajectories (Sec.~\ref{sec:methods}) use adaptive pseudo-time stepping (Appendix Table~\ref{tab:pf_sim}), so stored frames are not uniformly spaced in the applied displacement $u_y=t\,u_r$ (equivalently pseudo-time); the base increment $\Delta u_y\approx5\,\mu\mathrm{m}$ shrinks adaptively during rapid fracture.
The CNN2D--ConvGRU model nevertheless advances one frame index at a time: each call reconstructs the current checkpoint from $\mathcal{T}$ and sparse current-step measurements, without embedding a fixed $\Delta t$ in the network \cite{shi2015,ballas2016}.
In that sense the map is indexed on simulation checkpoints rather than on a prescribed clock, while path dependence in the training data is retained through the recurrent state and the history-field reference trajectories governed by Eq.~\eqref{eq:history}.
A single trained model can therefore follow trajectories whose FEM save times differ across cases.
This helps explain why Fig.~\ref{fig:damage_ratio} tracks the reference well during nucleation and coalescence before gradual drift under long sequential chains.

Deployment imposes two explicit inputs that should be stated as limitations rather than hidden assumptions.
The model cannot start from the boundary loading alone: it needs the past-history input $\mathcal{T}$ to get going, usually taken from the first few steps of a high-fidelity run or from an earlier sequential deployment (Fig.~\ref{fig:input_output}).
At the current step it also requires partial observations in the form of coarse $32\times32$ sparse displacement measurements, while the present-phase damage channel is withheld so that $\phi$ must be reconstructed.
Sparse displacement measurements do not uniquely determine the complete phase-field and displacement state in general.
The reconstruction performed by the network therefore depends on the temporal history and on statistical regularities learned from the finite-element training ensemble.
The reported test results demonstrate empirical reconstruction accuracy within this data distribution, but they do not establish uniqueness or observability for arbitrary crack configurations, loading conditions, or sensor layouts.
This matches the sparse and partial measurement setting emphasized in the Methods section.
This is a different information regime from dense field regression or physics-informed networks that ingest full state snapshots or PDE residuals at every step \cite{lee2020,raissi2019}.
Hybrid online correction with occasional finite element updates, or additional channels such as boundary reaction force, are natural ways to reduce long-horizon drift without abandoning the sparse-sensing interface.

\section{Conclusions}\label{sec:conclusions}

We developed a convolutional--recurrent model, i.e., CNN2D--ConvGRU, for measurement-conditioned inference of time-dependent phase-field brittle fracture from sparse and partial observations.
The model treats irreversible fracture evolution as reconstruction of the current full-field state rather than as regression to scalar failure quantities.
Given a fixed-length history $\mathcal{T}$ and sparse, partial measurements at the current load step, it reconstructs the full damage and displacement fields; sequential deployment assimilates sparse measurements at every load step.
%

The trained model reproduces the main crack paths, damage evolution, and bulk displacement response on report-set deployments.
The largest field errors are concentrated near actively propagating crack tips, where small differences in crack advance produce localized displacement and damage discrepancies.
The mean phase-field damage $\bar{\phi}$ tracks the reference trajectories over full sequential deployments at the training raster, although late-time drift can accumulate as one-step reconstruction errors compound.
Empirically, the same trained weights also transfer from $256\times256$ to $512\times512$ without retraining, using a proportionally refined $64\times64$ measurement grid.
We interpret this as a simultaneous raster-and-sensing generalization test on the same FEM discretization, not as a theoretical resolution-invariance property.

The model provides more than two orders of magnitude speedup relative to repeated finite element solves for the tested fracture trajectories.
Controlled architecture comparisons further show that CNN2D--ConvGRU offers an effective accuracy--cost compromise in this setting.
It outperforms temporal alternatives that discard spatial hidden states.
It also matches or exceeds FNO2D and PointNet spatial encoders while remaining cheaper than FNO2D, so Fourier operator layers are not useful here.
These results suggest that convolutional sequence learning can recover full-field phase-field fracture states when only sparse and partial observations are available at the current step.

Several limitations remain.
The current model requires an initial history window, uses sparse displacement measurements at the current step, and is trained on a modest two-dimensional fracture ensemble.
Long-horizon displacement accuracy is also more sensitive to sequential error accumulation than the phase-field damage prediction.
Future work may incorporate stronger physics constraints, richer loading and boundary-condition inputs, larger and more diverse training ensembles, three-dimensional fracture problems, and hybrid correction strategies that periodically couple the network with high-fidelity finite element updates.

\section*{Acknowledgment}

H.Z. acknowledged support from the Graduate Fellowship via Stanford University School of Engineering.
Z.Z. acknowledged support from Stanford Energy Fellowship, Precourt Institute for Energy, and Department of Chemical Engineering at Stanford University.

\section*{Data availability}

The code used in this study will be released on \url{https://github.com/hanfengzhai/PhFF-ML} upon acceptance of the manuscript.

\section*{Appendix}
\appendix
\counterwithin{equation}{section}
\counterwithin{figure}{section}
\counterwithin{table}{section}

\section{Phase-field simulation and data-set parameters}\label{sec:app_pf_sim}

Table~\ref{tab:pf_sim} collects every numerical setting used by the FEniCS reference solver and by the randomized crack ensemble that generates the training trajectories.

\begin{table}[ht]
    \centering
    \caption{Phase-field fracture simulation and training-ensemble parameters for the FEniCS reference solver.}
    \label{tab:pf_sim}
    \begin{tabular}{lll}
        \toprule
        Group & Parameter & Value \\
        \midrule
        \multirow{4}{*}{\makecell[l]{Geometry \&\\ discretization}}
          & Domain $\Omega$ & $[-0.5,0.5]^2$ ($1\,\mathrm{m}\times1\,\mathrm{m}$) \\
          & FEM spaces & CG1 for $\mathbf{u}$, $\phi$ (FEniCS) \\
          & Mesh resolution & $300$ cells per dimension ($\approx 3.3\,\mathrm{mm}$ element) \\
          & Precrack band width & $2.5\times$ cell size ($\approx 8.3\times10^{-3}$, $8.3\,\mathrm{mm}$) \\
        \midrule
        \multirow{6}{*}{\makecell[l]{Material \&\\ model}}
          & Critical fracture energy $G_c$ & $0.5$ \\
          & Regularization length $\ell$ & $10^{-2}$ ($10\,\mathrm{mm}$) \\
          & Lam\'e parameters $\lambda$, $\mu$ & $10^{6}$ \\
          & Degradation function & $g(\phi)=(1-\phi)^2$ \\
          & Elastic energy split & Volumetric--deviatoric (Eq.~\eqref{eq:psi_split}) \\
          & Irreversibility & History field $H$ (Eq.~\eqref{eq:history}) \\
        \midrule
        \multirow{6}{*}{\makecell[l]{Loading \&\\ BCs}}
          & Bottom edge & $\mathbf{u}=\mathbf{0}$ \\
          & Top edge & $u_y=t\,u_r$, traction-free otherwise \\
          & Displacement ratio $u_r$ & $5\times10^{-3}$ ($=5\,\mathrm{mm}$) \\
          & Applied-displ.\ increment $\Delta u_y=\Delta t\,u_r$ & $\approx 5\,\mu\mathrm{m}$ (steady; $0.25\,\mathrm{mm}$ first step) \\
          & Max applied displacement $u_y$ & $2.5\,\mathrm{mm}$ (nominal strain $0.25\%$) \\
          & Pre-existing cracks & $\phi=1$ on seeded subdomains \\
        \midrule
        \multirow{4}{*}{\makecell[l]{Solver \&\\ stepping}}
          & Scheme & Staggered alternate minimization \\
          & Staggered tolerance & $10^{-3}$ ($L^2$ on $\mathbf{u}$, $\phi$) \\
          & Initial pseudo-time step $\Delta t$ & $5\times10^{-2}$ (adaptive reduction) \\
          & Maximum pseudo-time & $0.5$ \\
        \midrule
        Termination
          & Stop criterion & $F_y<100$ for $10$ consecutive converged steps \\
        \midrule
        \multirow{5}{*}{\makecell[l]{Crack\\ ensemble}}
          & Ensemble size & $100$ independent random crack patterns \\
          & Cracks per simulation & $5$--$10$ (uniform random) \\
          & Crack length & Uniform on $[0.2,0.4]$ ($200$--$400\,\mathrm{mm}$) \\
          & Boundary clearance & $\ge 0.05$ ($\ge 50\,\mathrm{mm}$) from domain edges \\
          & Exported raster grids & $256\times256$ ($\approx 3.9\,\mathrm{mm}$/px); $512\times512$ ($\approx 1.95\,\mathrm{mm}$/px) \\
        \bottomrule
    \end{tabular}
\end{table}

The domain and CG1 discretization in Table~\ref{tab:pf_sim} match the weak forms in Eq.~\eqref{eq:weak_u}--Eq.~\eqref{eq:weak_phi}.
The unit square represents a $1\,\mathrm{m}\times1\,\mathrm{m}$ specimen.
The mesh is uniform with $300$ cells per side.
The regularization length $\ell=10^{-2}$ ($10\,\mathrm{mm}$) is then resolved by several elements across a diffuse crack band.
Material parameters follow a standard brittle phase-field calibration: quadratic degradation $g(\phi)=(1-\phi)^2$, volumetric--deviatoric elastic splitting, and irreversibility through the history field $H$ in Eq.~\eqref{eq:history}.
Loading is monotonic tension with adaptive pseudo-time stepping.
Each converged frame advances the applied top displacement by $\Delta u_y\approx5\,\mu\mathrm{m}$ in the steady regime, up to $u_y\le2.5\,\mathrm{mm}$.
Termination occurs when the force diagnostic $F_y$ in Eq.~\eqref{eq:reaction} stays below $100$ for ten consecutive converged increments.
The ensemble block defines how the $100$ crack patterns are sampled (crack count, length, and edge clearance).
It also lists the export rasters: $256\times256$ for deployment and $512\times512$ for the transfer study in Section~\ref{sec:scaling}.

\section{Training and loss hyperparameters}\label{sec:app_ml_train}

Table~\ref{tab:ml_train} lists the optimization, loss, and normalization settings for the reported CNN2D--ConvGRU model.

\begin{table}[ht]
    \centering
    \caption{Training, loss, and normalization hyperparameters for the reported CNN2D--ConvGRU model.}
    \label{tab:ml_train}
    \begin{tabular}{lll}
        \toprule
        Group & Hyperparameter & Value \\
        \midrule
        \multirow{6}{*}{Optimization}
          & Optimizer & Adam \\
          & Learning rate & $10^{-3}$ \\
          & Weight decay & $0$ \\
          & Batch size & $2$ \\
          & Physics regularization $\lambda_{\mathrm{phys}}$ & $0$ \\
          & FNO displacement refiner & off \\
        \midrule
        \multirow{4}{*}{Schedule}
          & Total epochs & $100$ \\
          & Mode & $\phi$-then-displacement \\
          & Stage epochs ($\phi$ / disp / fine-tune) & $30$ / $60$ / $10$ \\
          & Random seed & $123$ \\
        \midrule
        \multirow{5}{*}{\makecell[l]{Masks \&\\ loss}}
          & Mask type & sigmoid \\
          & Threshold $\tau$ & $1.5$ \\
          & Softness $\beta$ & $0.05$ \\
          & Interface band $b$ & $0.05$ \\
          & Interface weight $\gamma_{\mathrm{int}}$ & $3.0$ \\
        \midrule
        \multirow{3}{*}{Normalization}
          & Phase shift & $1.0$ \\
          & Displacement scale & $s_k=\max(|u_y^{\mathrm{top}}(k)|,\epsilon)$ (Eq.~\eqref{eq:disp_norm}) \\
          & $|u|$ input scale & $50.0$ (training-set only) \\
        \midrule
        \multirow{2}{*}{Data}
          & Training & $75$ trajectories \\
          & Report / model selection & $25$ held-out trajectories \\
        \bottomrule
    \end{tabular}
\end{table}

Optimization in Table~\ref{tab:ml_train} uses Adam with learning rate $10^{-3}$ and batch size $2$; $\lambda_{\mathrm{phys}}=0$ and the FNO displacement refiner are off for this reported checkpoint.
Adam was selected after the optimizer comparison in Table~\ref{tab:optimizer_sweep} using report-set displacement quality and wall-clock cost; Muon remains competitive, while SGD stays finite but about one to two orders worse on the displacement losses.
Training follows a three-stage schedule: $30$ epochs on the phase-field branch alone, $60$ epochs on the displacement branches with mask-weighted loss Eq.~\eqref{eq:loss_u}, and $10$ epochs of joint fine-tuning ($100$ epochs total).
The mask block defines the sigmoid solid mask, interface band, and interface weight $\gamma_{\mathrm{int}}$ in Eq.~\eqref{eq:masks}--Eq.~\eqref{eq:loss_u}.
Normalization entries match the shifted phase $\phi'$ and the prescribed-load displacement scaling in Eq.~\eqref{eq:disp_norm}.
The split used throughout the paper is $75$ training and $25$ report trajectories (Appendix~\ref{sec:app_benchmark}).
The checkpoint with lowest $\mathcal{L}_\phi+\mathcal{L}_u$ on the report set is used for all reported deployments.

\section{Architecture benchmark protocol}\label{sec:app_benchmark}

This appendix records the settings behind Table~\ref{tab:benchmark_combined}.
The benchmark compares alternative spatiotemporal backbones under the same data split, losses, sensing setup, hardware, and parameter budget.
We either swap the spatial encoder--decoder while keeping ConvGRU fixed, or swap the temporal operator while keeping the CNN2D stack fixed; in both cases the three-network layout of Section~\ref{sec:methods} is unchanged.

\paragraph{Data splits.}
The ensemble contains $100$ independent crack-pattern simulations on the $256\times256$ grid of Appendix Table~\ref{tab:pf_sim}.
Splits are defined by trajectory: each full loading history belongs to training or reporting, and every current-step sample from that simulation carries the same label.
Seventy-five trajectories are used for training and twenty-five for monitoring convergence, selecting checkpoints, ranking architectures, and reporting the scores in Table~\ref{tab:benchmark_combined}.
Report-set trajectories are not used for parameter updates; the same fixed split is applied to every architecture.

\paragraph{Suites and swapped components.}
The \emph{spatial} suite replaces the CNN2D encoder--decoder while keeping the ConvGRU cell: CNN2D--ConvGRU (baseline), FNO2D--ConvGRU, DeepONet--ConvGRU, and PointNet--ConvGRU.
The \emph{temporal} suite keeps the CNN2D encoder--decoder and swaps the convolutional recurrent core for a flat GRU, LSTM, vanilla RNN, or neural ODE integrator: CNN2D--ConvGRU (baseline), CNN2D--GRU, CNN2D--LSTM, CNN2D--RNN, and CNN2D--ODE.
All candidates share the same loss functions, masks, and normalization settings (Table~\ref{tab:ml_train}).

\paragraph{Matched parameter budget.}
To keep model size from confounding the comparison, every non-baseline candidate is tuned to within $\pm5\%$ of the baseline trainable parameter count.
The reference budget includes all weights in the three prediction networks (FNO displacement refiner off, matching the production checkpoint).
The width of each alternative architecture is adjusted automatically until its parameter count matches this budget; the baseline itself is not width-tuned.

\paragraph{Evaluation metrics.}
After training, each model is scored in current-step reconstruction mode on the report trajectories only (no multi-step sequential deployment in this table).
At each report-set step the network receives the past-history window $\mathcal{T}$ and sparse current-step measurements and reconstructs the current fields; predictions are compared to the finite element target on the full $256\times256$ grid.
The phase field is scored in physical units $\phi\in[0,1]$ after subtracting the training shift.
Table~\ref{tab:benchmark_combined} reports the mean $R^2_{\phi}$ averaged over the $25$ report-set crack patterns; a model passes if $R^2_{\phi}\ge0.99$.
Wall-clock training time is recorded on a single NVIDIA L40S GPU and includes the full $100$-epoch staged schedule.

Table~\ref{tab:benchmark_protocol} collects the shared hyperparameters.
Per-simulation scores for $u_x$ and $u_y$, and epoch-wise training losses, are archived with the benchmark outputs but are not used in Table~\ref{tab:benchmark_combined}.

\begin{table}[ht]
    \centering
    \caption{Shared protocol for the architecture benchmark suites (Section~\ref{sec:results}, Table~\ref{tab:benchmark_combined}).}
    \label{tab:benchmark_protocol}
    \begin{tabular}{lll}
        \toprule
        Group & Setting & Value \\
        \midrule
        \multirow{3}{*}{\makecell[l]{Data\\ splits}}
          & Training trajectories & $75$ \\
          & Report trajectories & $25$ \\
          & Past-history length $|\mathcal{T}|$ & $10$ \\
        \midrule
        \multirow{4}{*}{Sensing}
          & Sparse sensor grid & $32\times32$ \\
          & Phase input & shifted $\phi'$ and sparse $|u|$ \\
          & Displacement input & normalized $u_x$, $u_y$, $\phi'$ history + sparse current step \\
          & Output channels & $\phi'$, normalized $u_x$, $u_y$ \\
        \midrule
        \multirow{5}{*}{Optimization}
          & Optimizer / learning rate & Adam / $10^{-3}$ \\
          & Batch size & $2$ \\
          & Total epochs & $100$ \\
          & Schedule ($\phi$ / disp / fine-tune) & $30$ / $60$ / $10$ \\
          & Physics regularization $\lambda_{\mathrm{phys}}$ & $0$ \\
        \midrule
        \multirow{3}{*}{\makecell[l]{Masks \&\\ loss}}
          & Mask type / threshold & sigmoid / $1.5$ \\
          & Softness $\beta$ / band $b$ & $0.05$ / $0.05$ \\
          & Interface weight $\gamma_{\mathrm{int}}$ & $3.0$ \\
        \midrule
        \multirow{1}{*}{Architecture}
          & Hidden channels $C_h$ & $32$ \\
        \midrule
        \multirow{1}{*}{\makecell[l]{Baseline\\ refiner}}
          & Displacement refiner (baseline) & off \\
        \midrule
        Capacity matching & Parameter budget & within $\pm5\%$ of baseline ($\approx2.52\times10^{5}$) \\
        \midrule
        Hardware & GPU & NVIDIA L40S (single device) \\
        \midrule
        \multirow{3}{*}{Evaluation}
          & Inference mode & current-step reconstruction \\
          & Primary metric & mean $R^2_{\phi}$ on report trajectories \\
          & Pass criterion & $R^2_{\phi}\ge0.99$ \\
        \bottomrule
    \end{tabular}
\end{table}

Table~\ref{tab:benchmark_protocol} fixes the protocol behind Table~\ref{tab:benchmark_combined}.
It uses a $75$/$25$ trajectory split, $|\mathcal{T}|=10$, $32\times32$ sparse sensing, Adam at $10^{-3}$, and the same $30$/$60$/$10$ staged schedule and mask settings as Table~\ref{tab:ml_train}.
Every candidate is width-matched to within $\pm5\%$ of the baseline parameter count ($\approx2.52\times10^{5}$) with the FNO displacement refiner off.
Capacity differences therefore do not explain the pass/fail outcomes.
Evaluation is current-step reconstruction on the report set only.
A model passes if mean $R^2_{\phi}\ge0.99$ on that set.
Wall-clock training time is recorded on a single NVIDIA L40S.

\section{Conditional sequential-reconstruction error estimate}
\label{app:sequential_error}

The deployed model reconstructs the current state from a history window and
an external sparse measurement. Let
\[
Z_k=(X_{k-m},\ldots,X_{k-1})
\]
denote the length-$m$ reference history at load step $k$, and let
\[
\widehat{Z}_k
=
(\widehat{X}_{k-m},\ldots,\widehat{X}_{k-1})
\]
denote the corresponding history used by the network. For the same
current-step measurement $y_k$, write
\[
X_k=S_k(Z_k,y_k),
\qquad
\widehat{X}_k=\widehat{S}_k(\widehat{Z}_k,y_k),
\]
where the maps may depend explicitly on the load index $k$.

For a history $Z=(Z^{(1)},\ldots,Z^{(m)})$, define the maximum history norm
\[
\|Z\|_{\infty,m}
=
\max_{1\leq j\leq m}\|Z^{(j)}\|.
\]

\begin{proposition}[Conditional sequential-reconstruction error estimate]
Assume that the reference and reconstructed histories remain in a set
$\mathcal{D}$ on which
\[
\|S_k(Z,y_k)-S_k(\widetilde{Z},y_k)\|
\leq
L\|Z-\widetilde{Z}\|_{\infty,m}
\]
for all admissible histories $Z,\widetilde{Z}\in\mathcal{D}$. Assume also
that the one-step approximation error satisfies
\[
\|S_k(Z,y_k)-\widehat{S}_k(Z,y_k)\|
\leq
\varepsilon_k
\]
for all $Z\in\mathcal{D}$.

Define
\[
e_k=\|X_k-\widehat{X}_k\|,
\qquad
E_k=\max_{k-m\leq j\leq k-1}e_j,
\qquad
M=\max\{1,L\}.
\]
Then
\[
e_k\leq L E_k+\varepsilon_k
\]
and
\[
E_{k+1}\leq M E_k+\varepsilon_k.
\]
Consequently, for any integer $n\geq1$,
\[
E_{k+n}
\leq
M^nE_k
+
\sum_{r=0}^{n-1}M^{n-1-r}\varepsilon_{k+r}.
\]
If the initial history is exact, $E_m=0$, and
$\varepsilon_k\leq\varepsilon$, then
\[
E_{m+n}
\leq
\begin{cases}
\displaystyle
\varepsilon\frac{M^n-1}{M-1},
& M>1,\\[8pt]
n\varepsilon,
& M=1.
\end{cases}
\]
\end{proposition}

\begin{proof}
Using the triangle inequality and the two assumptions,
\[
\begin{aligned}
e_k
&=
\|S_k(Z_k,y_k)-\widehat{S}_k(\widehat{Z}_k,y_k)\|\\
&\leq
\|S_k(Z_k,y_k)-S_k(\widehat{Z}_k,y_k)\|
+
\|S_k(\widehat{Z}_k,y_k)
-\widehat{S}_k(\widehat{Z}_k,y_k)\|\\
&\leq
L\|Z_k-\widehat{Z}_k\|_{\infty,m}
+\varepsilon_k\\
&=
LE_k+\varepsilon_k.
\end{aligned}
\]
After the newly reconstructed state is appended and the oldest state is
removed, the next history error satisfies
\[
E_{k+1}
=
\max_{k-m+1\leq j\leq k}e_j
\leq
\max\{E_k,e_k\}.
\]
Therefore,
\[
E_{k+1}
\leq
\max\{E_k,LE_k+\varepsilon_k\}
\leq
ME_k+\varepsilon_k,
\]
where $M=\max\{1,L\}$. Iterating this recurrence gives
\[
E_{k+n}
\leq
M^nE_k
+
\sum_{r=0}^{n-1}M^{n-1-r}\varepsilon_{k+r}.
\]
The uniform-error result follows by evaluating the geometric sum.
\end{proof}

This estimate is conditional and is used only to illustrate how errors in
the complete history window may propagate during sequential reconstruction.
It does not establish global Lipschitz continuity of the phase-field fracture
evolution or of the trained network.

\printbibliography
\end{document}